\documentclass[article]{IEEEtran}

\usepackage{cite}
\usepackage{amsmath,amssymb,amsfonts}
\usepackage{graphicx}
\usepackage{textcomp}
\usepackage{xcolor}
\def\BibTeX{{\rm B\kern-.05em{\sc i\kern-.025em b}\kern-.08em
    T\kern-.1667em\lower.7ex\hbox{E}\kern-.125emX}}
\usepackage{cite}
\usepackage[square,numbers]{natbib}
\usepackage{amsmath}
\usepackage{amssymb}
\usepackage{amsfonts}
\usepackage{amsthm}
\usepackage{algorithm,algcompatible,amsmath}
\usepackage{color}
\usepackage{comment}
\usepackage{amssymb}
\usepackage{graphicx}
\usepackage{lipsum}
\usepackage[colorinlistoftodos]{todonotes}
\usepackage{textcomp}
\usepackage{xcolor}

\numberwithin{equation}{section}

\newtheorem{thm}{Theorem}[section]
\newtheorem{prop}[thm]{Proposition}

\newtheorem{remark}[thm]{Remark}

\newtheorem{lem}[thm]{Lemma}

\algnewcommand\INPUT{\item[\textbf{Input:}]}
\algnewcommand\OUTPUT{\item[\textbf{Output:}]}

\begin{document}
\title{CoDGraD: A Code-based Distributed Gradient Descent Scheme for Decentralized Convex Optimization
\thanks{}
}

\author{\IEEEauthorblockN{Elie Atallah, Nazanin Rahnavard, and Qiyu Sun}\\
\IEEEauthorblockA{\textit{Department of Electrical and Computer Engineering, University of Central Florida, Orlando, FL} \\
\textit{Department of Mathematics, University of Central Florida, Orlando, FL}\\
\textit{Emails: elieatallah@knights.ucf.edu, nazanin@ece.ucf.edu, Qiyu.Sun@ucf.edu}
}}
\maketitle

\begin{abstract}  In this paper, we consider  a large network containing many regions such that each region is equipped with a worker with some data processing
and communication capability.
For such a  network, some workers may become stragglers
due to the failure or heavy delay on computing or communicating.
To resolve the above straggling problem, a coded scheme that introduces
certain redundancy for every worker was recently proposed, and a
gradient coding paradigm  was developed to solve  convex optimization problems
when the network has a centralized fusion center.
In this paper, we  propose an iterative distributed algorithm,
referred as 
{Code-Based Distributed Gradient Descent algorithm} (CoDGraD), to solve  convex optimization problems over distributed networks.
In each iteration of the proposed algorithm,  an active worker   shares the coded local gradient and  approximated solution of the convex optimization problem with
non-straggling workers at the adjacent regions only.
In this paper, we also provide the consensus and convergence analysis for the CoDGraD algorithm and
 we demonstrate its performance via numerical simulations. 
\end{abstract}


\section{Introduction}
\IEEEPARstart{C}{onvex} optimization on a  network  of large size 
 has played a significant role for solving various problems, such as big-data processing in machine learning,
 distributed parameter estimation in wireless sensor networks,
 distributed sampling and signal reconstruction, distributed design of filter banks,
 distributed spectrum sensing in cognitive radio
networks,  source localization in cellular networks
\cite{nedic2001, dasilva2007, johansson2008, nedic2010, fu2015,  cheng2017, jiang2017}.
The objective functions $f$ in  such
 optimization problems,
\begin{equation}\label{globalobjectivefunction}
 f({\bf x}) = \sum_{l=1}^m f_{l}({\bf x}),  \ {\bf x}\in \mathbb{R}^N,
\end{equation}
  are often  the summation of  some local objective functions  $f_{l}, 1\le l\le  m$,  related to a partition
 of the network. In this paper, we consider  
 the scenario that
  each region of the partition is equipped with
a worker that  has some data processing and communication ability, while the  network has a fusion center with limited 
capacity
or it  does not have a fusion center at all.

For a network with centralized data processing facility,
the following optimization problem
\begin{equation}\label{optimizationproblem}
\bar{\bf x}  
= \arg\min_{{\bf x}\in \mathbb{R}^{N}}\sum_{l=1}^m f_{l}({\bf x}),
\end{equation}
associated with the objective function $f=\sum_{l=1}^m f_l$
in \eqref{globalobjectivefunction}
has been well studied,   see
  \citep{nedich2015, bedi2018, gurbuzbalaban2015, emirov2021, nedic2001, nedic2010,  takahashi2009,  cattivelli2010,  bertsekas2011,  sayed2013, needell2014, sayed2014}
and references therein for various algorithms implementable in a strong  fusion center or in local workers distributed over the network.
 Denote  the gradient of $g$ on ${\mathbb R}^N$ by $ \nabla g $.
For a network equipped with one data processing unit only, a conventional approach to  the optimization problem \eqref{optimizationproblem} is the  \emph{gradient descent} algorithm,
\begin{equation}\label{classicalgradient}
    {\bf x}(k+1)={\bf x}(k)-\frac{\alpha_{k}}{m}\sum_{l=1}^m \nabla{f_l}({\bf x}(k)), \ k\ge 0,
\end{equation}
where
$\{\alpha_k\}_{k=0}^\infty$ is a positive sequence chosen appropriately \citep{mathews1993, robbins1951, schraudolph1999, friedman2002, bottou2012,  zeiler2012,
 kingma2014, hardt2015, tan2016}.
Our illustrative examples of step sizes $\alpha_k, k\ge 0$, are
\begin{equation}\label{stepsize.example}
\alpha_k=(k+a)^{-\theta}, \ k\ge 0,\end{equation}
 for some  $\theta\in (1/2, 1)$ and $a\ge 1$.

Several distributed versions of the  gradient descent  algorithm  \eqref{classicalgradient}
have been proposed, including
the Adapt-Then-Combine  algorithm  (ATC)  and  the Combine-Then-Adapt algorithm (CTA)   \citep{cattivelli2010,sayed2013},
where implementation of the worker
at  agent  $l$ is given by
\begin{equation}\label{ATC}
\Big\{\begin{array}{l}
  {\bf y}_{l}(k) = {\bf x}_{l}(k) - \alpha_{k} \nabla{f}_{l}({\bf x}_{l}(k)) \\
  {\bf x}_{l}(k+1) = \sum_{ j \in \mathcal{N}_{l} }w_{lj}(k) {\bf y}_{j}(k)
\end{array}
\end{equation}
and
\begin{equation} \label{CTA}
\Big\{\begin{array}{l}
  {\bf y}_{l}(k) = \sum_{ j \in \mathcal{N}_{l} }w_{lj}(k) {\bf x}_{j}(k) \\
  {\bf x}_{l}(k+1) = {\bf y}_{l}(k) - \alpha_{k} \nabla{f}_{l}({\bf y}_{l}(k))
\end{array}
\end{equation}
respectively, where ${\mathcal N}_l$ contains all adjacent agents of the agent $l$ for data sharing,
and
 $(w_{lj})_{1\le l, j\le m}$ is a consensus matrix. 
The above  ATC and CTA algorithms  
may reach consensus over all nodes to the  optimal solution  $\bar {\bf x}$ in
\eqref{optimizationproblem}.
They are essentially the same as
 the gradient descent  algorithm  \eqref{classicalgradient} if
the graph to describe  topological structure of  the network is complete and
the consensus weights $w_{lj}=1/m, 1\le l, j\le m$.
However, comparing with the gradient descent  algorithm  \eqref{classicalgradient},
in the  implementation of  the  ATC and CTA algorithms, we circumvent
the expansive evaluation of gradient  $\nabla f$ of the global objective function by
evaluating  gradients  $\nabla f_l, 1\le l\le m$, of  local objective functions  at each worker node
and then communicating local gradients to the neighbors with nonzero consensus weights.

In applications such as distributed learning and optimization over the cloud,
some workers in the network may  become inactive 
 due to the failure or  heavy  delay on computing or  communicating  
\cite{dean2012,  ho2013, li2014}.
 To resolve the  above problem,
  uncoded and coded  local gradients  have been proposed in
 \citep{ tandon2017, halbawi2017,
raviv2017} to
recover the full gradient  from  local gradients
on  active nodes  
$ \Delta \subset\{1, \ldots, m\}$.
Without loss of generality, we assume that $\Delta \subset \{1, \ldots, n\}$, where $n\le m$ is the number of active nodes.
In this paper, we follow the coded scheme in \citep{tandon2017}  and consider the  paradigm
that for every active node  $i$,   the global objective function can be recovered from coded objective functions
$g_{j}$ on non-stragglers relative to node $ i $,  i.e.,
 \begin{equation}\label{fgi.eq}
 \begin{split}
 f({\bf x}) & = \sum_{j=1}^{n}  a(i,j) g_{j}({\bf x})
 \end{split}
 \end{equation}
 for some decoding matrix
$ {\bf A}=(a(i,j))_{1\le i, j \leq n}$.
 The above requirement is met if
  the coded scheme for non-stragglers is given by
  \begin{equation} \label {Bfit.def}
  g_i({\bf x})=\sum_{l=1}^{m} b(i, l) f_l({\bf x}),\  1\le i\le n,
  \end{equation}
  and
  the coding matrix
 ${\bf B}=(b(i,l))_{1\le i \leq n, 1\le l \leq m}$
  satisfies
     \begin{equation}\label{AB1.eq}
  {\bf A} {\bf B}={\bf 1}_{n \times m},
  \end{equation}
 where   ${\bf 1}_{n \times m}$ is the $n \times m$ matrix with all entries  taking value $ 1 $. \\

The coded scheme \eqref{fgi.eq} and \eqref{Bfit.def}
has been used in \cite{tandon2017} to solve the global optimization problem \eqref{optimizationproblem},
where  the worker at each region  evaluates the coded local gradients and sends them to  the worker at the master node; then the worker at the master node
aggregates a weighted sum of coded local gradients to form the gradient of
the global objective function, and applies a centralized gradient descent approach similar to \eqref{classicalgradient}
to update the approximation  to the optimal solution $\bar {\bf x}$ in \eqref{optimizationproblem};
and finally the worker at the master node sends the updated approximation to  all local workers on the network for the next iteration.
In this paper, based on the coded  scheme \eqref{fgi.eq} and \eqref{Bfit.def},
 we propose a \emph{distributed} algorithm  to solve the   
 convex optimization problem \eqref{optimizationproblem}.

\subsection{Objectives and Contributions}

Since the focus of distributed optimization is overwhelmingly occupied with first-order methods, it is worth trying to enhance such methods rather than employing methods of another type.  Our initial aim of this paper
is to improve the performance of distributed gradient descent  algorithm through utilizing coding.
In particular, we  focus  on the leveraging of coding on the performance of the distributed optimization algorithm and how the convergence rate of these algorithms can be better enhanced through using an appropriate coding scheme based on the network topology.

While distributed optimization is implementable through distributing the local functions among the nodes to sum up to the global optimization function as in \eqref{globalobjectivefunction}, we follow here an alternative route. We decompose the global function among the nodes in a coded manner and try to implement an algorithm which allows optimization under such scheme \eqref{fgi.eq}.
To this end, we adapt the stochastic form
of the decoding matrix ${\bf{A}} $
 in our proposed algorithm,
  carry the negativity in some of its coefficients to the gradient operation, and then utilize  gradient descent and  ascent steps, see \eqref{Afit.def} and \eqref{mainalgorithm}.

In  this paper, we do not impose any algebraic structure of
the coding matrix ${\bf B}$ and  the decoding matrix ${\bf A}$,
and thus our approach applies for  a broader class of coding/decoding matrices as long as the allowed number of stragglers is fulfilled, except that the normalized decoding matrix $|\tilde {\bf A}|$ in
\eqref{stochasticdecodematrix.def0}
is assumed to have simple eigenvalue one, see   
Proposition \ref{eigenvalueone.pr}.
Our work has the assumptions of strongly convex global function $f$, Lipschitz continuous (coded) gradients $ \nabla{g}_{i} $, and uniformly bounded (coded) gradients $ \nabla{g}_{i} $.
We believe that
if the coding/decoding matrices have
additional structure properties \cite{yuan2016},
our  algorithm  could converge
under weak requirements on the objective functions and coded gradients.

In this paper, we utilize  exact gradient coding  for fusion centralized networks and  investigate the implications of such coding  used in a multi-agent setting. This serves well to  our initial aim of showing how enhancement to the convergence rate can be accomplished, see the consensus and convergence conclusions in Sections
 \ref{mainalg1.section} and
\ref{mainalg2.section} of our proposed   algorithm.
The paradigm of approximate gradient coding is discussed in \cite{glasgow2020}. It could be an interesting problem to extend our convergence conclusions  in the above setting.

In this paper, we work on distributed optimization problem for multi-agent systems on fixed static network topology. Thus, the active nodes are the non-straggler nodes which are fixed throughout the proposed distributed algorithm and where coded local functions are used according to the coding scheme \eqref{fgi.eq} instead of the brute decomposition as in \eqref{globalobjectivefunction}. We postpone the consideration of time-varying straggler networks to a later endeavor.

The implementation of distributed algorithms on
static/time-varying networks with stragglers
is of importance.
We wish that this work may serve as a starting point for a full-fledged investigation to  distributed algorithms on static/time-varying networks with stragglers, with applications to the engineering field, such as federated decentralized learning and distributed machine learning.

    \subsection{Organization and notations}
In Section \ref{results}, we formulate our  code-based distributed gradient descent algorithm (CoDGraD) to solve the optimization problem \eqref{optimizationproblem}.
In Section~\ref{Network_Forming_Coding} we describe the coordinated distributed formation of the network and its data coding.
Then in Sections \ref{mainalg1.section} and
\ref{mainalg2.section}, we consider the consensus  and convergence properties of  the proposed  CoDGraD  algorithm.
Afterwards, we demonstrate the performance of CoDGraD through simulation in Section \ref{simulation.section}. We conclude the paper in Section \ref{conclusion.section}.

In what follows, we use bold capital letters, bold lower case letters and lower case letters for matrices, vectors and scalar variables, respectively.
Denote the positive part of a real number $t$ by
  $t_+=\max(t, 0)$,
Denote the matrix of  dimension $n\times m $ with all entries taking value one by
 $ {\bf 1}_{n\times m} $ (and ${\bf 1}_n$ when $m=1$) and the unit matrix of size $n\times n$ by ${\bf I}_{n}$.
 and the zero matrix of size $n\times m$ by  ${\bf 0}_{n\times m}$   respectively.
Denote the transpose of an matrix ${\bf A}$ by ${\bf A}^T$, and
the transpose and standard $\ell^p$-norm of a vector ${\bf x}$ by  ${\bf x}^T$  and $\|{\bf x}\|_p, 1\le p\le \infty$,
respectively.

\section{A Code-based Distributed Gradient Descent Algorithm}\label{results}

Let the coded objective functions $g_j, 1\le j\le n$, and  the decoding matrix ${\bf A}=(a(i,j))_{1\le i,j\le n}$ be as in the
coded scheme \eqref{fgi.eq} and \eqref{Bfit.def}.
Set decoding weights
\begin{equation}\label{weight.def} w_i= \Big(\sum_{j \in \Gamma_i} |a(i,j)|\Big)^{-1},\ 1\le i\le n,
\end{equation}
 where
\begin{equation}\label{gammai.def}
\Gamma_{i}=\{j, \  a(i,j)\ne 0\}.\end{equation}
In
\cite{Atallah2018ACD, tandon2017},  all workers not in $\Gamma_i$ are considered as ``stragglers"  or ``non-neighboring workers"  for an active node $i$, since coded information
at those workers are not used to evaluate the gradient  $\nabla f$ of the global objective function $f$ at the node $i$. We remark that in
this paper active nodes are those workers over the network that don't witness delay or failure on computing or communicating.
Being in a fixed static network implementation then all nodes in the graph of the network topology are considered as active nodes.

To solve the  optimization problem \eqref{optimizationproblem}, we propose an iterative distributed algorithm
 with  the implementation of the worker at the region $i$ given by   
\begin{eqnarray}\label{mainalgorithm}
\left\{\begin{array}{l}
{\bf v}_i(k)= \nabla g_i( {\bf x}_i(k)),\\
{\bf y}_i^+(k)= {\bf x}_i(k)- \alpha_k {\bf v}_i(k) ,\\
{\bf y}_i^-(k)= {\bf x}_i(k)+\alpha_k {\bf v}_i(k),\\
\begin{array}{lcl}  \hskip-0.08in {\bf x}_i(k+1)   & \hskip-0.08in = &  \hskip-0.08in  \sum_{j\in \Gamma_i}  w_i
\big\{(a(i, j))_+ {\bf y}_j^+(k) \\
& &  \quad +  (- a(i, j))_+
  {\bf y}_j^-(k)\big\}, \ \ k\ge 0,
\end{array}
\end{array}
\right.
\end{eqnarray}
where initials ${\bf x}_i(0)$ are chosen randomly or set zero initially,  and step sizes  $\alpha_k, k\ge 0$,  \citep{robbins1951,bertsekas2011,bottou2012} are so chosen that 
\begin{equation}
\label{alphak.assumption}
\sum_{k=0}^\infty \alpha_k=\infty \ \ {\rm and }\  \ \sum_{k=0}^\infty \alpha_k^2<\infty.
\end{equation}
We call this algorithm as  a code-based distributed gradient descent algorithm and  use CoDGraD for abbreviation.
 The implementation of the CoDGraD is described in Algorithm 1.
\begin{algorithm}[t]
 \caption{The CoDGraD Algorithm}
 \begin{algorithmic}[1]
 \INPUT 
  Set tolerance value $\epsilon $ for halting the algorithm and the number of iteration $k=0$.
 Initialize  an estimate $ {\bf x}_{i}(0) $ of the optimal solution $\bar {\bf x}$  and approximation error $  e_{i}(0)=\epsilon$ at the worker $ i $.

  \WHILE{ $ e_{i}(k) \geq \epsilon $} \Comment{Halting is done at each node independently with no coordination}

        \STATE {At the worker $i$, use \eqref{mainalgorithm} to update the estimate ${\bf x}_i(k)$}.
        \STATE Find error $ e_{i}(k+1) = \| {\bf x}_{i}(k+1)-{\bf x}_{i}(k) \| $  for all $ 1 \leq i \leq n $
        \STATE $k=k+1$
    \ENDWHILE
    \OUTPUT ${\bf x}_{i}(k)$  and $k$. 
  \end{algorithmic}
  \label{algorithm1.pseudocode}
\end{algorithm}

For our purpose, we consider the
coded scheme \eqref{fgi.eq} and \eqref{Bfit.def} that matches our network topology.
Here the network topology is described by
 an undirected graph  ${\mathcal G}=(V, E)$, where the worker in each region is represented by
a vertex in $V$ and each edge  $(l,l')\in E$ means that workers in the regions $l$ and $l'\in V$ have direct communication for data sharing.
 Therefore the topological matching property of the
coded scheme \eqref{fgi.eq} and \eqref{Bfit.def} is satisfied
if the decoding weight $a(i,j)$ takes zero value whenever there is no direct communication between active workers in the  region $i$ and $j$, i.e.,
$$a(i,j)=0 \ \ {\rm if} \ \ (i,j)\not\in E.$$
This means that any non-straggling node of any active node  $ i $  which is contributing in the decoding process must be a neighbor of the active node $ i $ in the whole network, i.e., $\Gamma_i\subset {\mathcal N}_i$,
where $\Gamma_i$ is given in  \eqref{gammai.def}.  However, the converse is not necessarily true as a non-straggler node may be used in the coding procedure and does not contribute to the active node $i$.

For the above scenario of the  coded scheme \eqref{fgi.eq} and \eqref{Bfit.def}, the active worker at each region
first evaluates the coded local gradient, then it updates its estimate through gradient descent/ascent step, next it shares the updated estimate approximations with neighboring active workers at the adjacent regions, henceforth updating its estimate towards the optimal solution $\bar {\bf x}$,
Hence the CoDGraD algorithm is implementable in the network that does not have a fusion center at all.

For the decoding matrix ${\bf A}$ 
in  \eqref{fgi.eq},
we define  its normalized decoding matrix ${\bf A}_{\rm nde}= \big( \tilde a(i,j)\big)_{1\le i, j\le n}$ of size $n\times n$
by
\begin{equation}\label{normalizedfgi.eq}\tilde a(i,j)= w_i a(i,j) ,\  1\le i, j\le n,
\end{equation}
\noindent and  its row stochastic decoding matrix ${\bf A}_{\rm sde}$ of size $2n\times 2n$  by
\begin{equation} \label{Afit.def}
 {\bf A}_{\rm sde}=\left(\begin{array}{cc}
\tilde {\bf A}_+  &  \tilde {\bf A}_{-} \\
\tilde {\bf A}_{+} &\tilde {\bf  A}_{-}
\end{array}\right),
\end{equation}
\noindent
 where   $w_i, 1\le i\le n$, are decoding weights in \eqref{weight.def}, and
 $$\tilde {\bf A}_+=\big((\tilde a(i,j))_+\big)_{1\le i,j\le n}\ \ {\rm  and}\  \ \tilde {\bf A}_-=\big((-\tilde a(i,j))_+\big)_{1\le i,j\le n}$$
\noindent
 are positive/negative parts of  the normalized decoding matrix $\tilde {\bf A}=(\tilde a(i,j))_{1\le i, j \leq n}$  respectively.

In this paper, we consider the {\bf consensus}  and {\bf convergence} properties of  ${\bf x}_i(k), k \ge 0$, in the proposed CoDGraD  algorithm  \eqref{mainalgorithm}
under the following assumptions:  (i)\ The row stochastic matrix $ {\bf A}_{\rm sde}$ in \eqref{Afit.def} has  simple eigenvalue one and all other eigenvalues contained in the open
unit complex disk centered at the origin; (ii)\  The global objective function $f$
 is a differentiable strongly convex; (iii)\ The (coded) local  objective functions $g_i, 1\le i\le n$,  have bounded gradients; and (iv) the (coded) local objective functions $g_j, 1\le j\le n$ are differentiable and have  continuous gradients.

\subsection{Conditions on the Network Topology and Coding Scheme}

The code-decode scheme in \eqref{AB1.eq} presented in this paper is essentially
the same as the scheme in \citep{tandon2017},
where $ n, m, s $ denote the number of workers, samples and
stragglers respectively.
Let $\Gamma_i, 1\le i\le n$, be given in \eqref{gammai.def} and
denote their complements by $\Gamma_i^c, 1\le i\le n$.
Therefore we require that stragglers to a node $ i $ are contained in $ \Gamma_{i}^{c} , 1 \leq i \leq n $.

In the following, we consider two scenarios
of network topology in which  the matrix $ {\bf A}_{\rm sde} $ used in the CoDGraD algorithm \eqref{mainalgorithm}
has simple eigenvalue one and all other eigenvalues contained in the open unit complex disk centered at the origin.
The first  scenario is of full
cyclic assignment of the  active workers after certain permutation, i.e.,
the decoding matrix
${\bf A}=(a(i,j))_{1\le i, j\le n}$
with   the first row
 having its first $n-s$ entries assigned as non-zero, and
 as we move down the rows, the positions of the $n-s$
non-zero entries shifting one step to the right, and cycle around
until the last row. Mathematically,  the decoding matrix ${\bf A}$
satisfies the following conditions:
\begin{equation} \label{Tandonlemma.eq0}
a(i,j)\ne 0\end{equation}
\noindent
for all $(i,j)$ satisfying either $i\le j\le \min(i+n-s, n)$
or $ j+s\le i$, cf., \citep[eq.10]{tandon2017}.
For the above scenario, we have

\begin{prop}\label{lemma_A_sde_simple_eig}
Let $1\le s\le n-1$ and
the decoding matrix
${\bf A}$ satisfy
\eqref{Tandonlemma.eq0}.
Then the matrix $ {\bf A}_{\rm sde} $ used in the CoDGraD algorithm \eqref{mainalgorithm}
has simple eigenvalue one and all other eigenvalues contained in the open unit complex disk centered at the origin.
\end{prop}

Denote the normalized decoding matrix of the decoding matrix ${\bf A}$
by
\begin{equation}\label{stochasticdecodematrix.def0}
    |\tilde {\bf A}|=(|\tilde a(i,j)|)_{1\le i, j \leq n},
\end{equation}
where $\tilde a(i,j), 1\le i,j\le n$, are given in \eqref{normalizedfgi.eq}.
By
\eqref{normalizedfgi.eq},
the normalized decoding matrix $|\tilde {\bf A}|$ has  row stochastic property.
To prove Proposition \ref{lemma_A_sde_simple_eig}, we need an equivalence
between the
eigenvalue properties for  the row stochastic  matrices ${\bf A}_{\rm sde}$ and  $|{\tilde {\bf A}}|$.

\begin{prop}\label{eigenvalueone.pr}  The algebraic multiplicities of nonzero eigenvalues of
 stochastic decoding matrices  ${\bf A}_{\rm sde}$ and $|{\tilde {\bf A}}|$ are the same.
\end{prop}

The proof of the above proposition will be given in  Appendix
\ref{eigenvalueone.pr.appendix}.  We assume that Proposition {eigenvalueone.pr}
holds and we  give the proof of Proposition
\ref{lemma_A_sde_simple_eig} below.

\begin{proof}[Proof of Proposition
\ref{lemma_A_sde_simple_eig}]
%
Set ${\bf B}=|\tilde {\bf A}|$, and write
$B^k=(b_k(i,j))_{1\le i, j\le n}, \ k\ge 1$, and also
${\bf B}=(b(i,j))_{1\le i,j\le n}$ for $k=1$.
Then ${\bf B}$ is a row stochastic matrix with nonzero diagonal entries.
By Perron-Frobenius theorem and Proposition \ref{eigenvalueone.pr}, it is suffices to prove that
 ${\bf B}$ is irreducible, i.e., for any $1\le i,j\le n$, there exists $k\ge 1$ such that
\begin{equation} \label{Tandonlemma.pfeq0}
b_k(i,j)\ne 0.
    \end{equation}
    By the assumption on the decode matrix ${\bf A}$, we have
   \begin{equation}\label{Tandonlemma.pfeq1}
b(i,j)\ne 0 \ {\rm if} \ i\le j\le i+n-s\ {\rm and\ if} \ j+s\le i.
\end{equation}
Observe that for any $1\le i, j\le n$, we have
\begin{eqnarray}
\label{Tandonlemma.pfeq1+}
b_{k+1}(i,j) & = & \sum_{l=1}^n b(i,l) b_k(l,j) \ge \sum_{l=i}^{\min(i+n-s, n)} b(i,l) b_k(i,j)\nonumber\\
&  &  +\sum_{l=1}^{i-s} b(i,l)b_k(l,j),\  k\ge 1.
\end{eqnarray}
    \noindent
By \eqref{Tandonlemma.pfeq1}
and \eqref{Tandonlemma.pfeq1+}, we can prove by induction on $k\ge 1$ that
$b_k(i,j)\ne 0$
 for all $(i,j)$ satisfying either
 $i\le j\le \min (i+k(n-s), n)$  or $j+ks\le  i$.
This proves \eqref{Tandonlemma.pfeq0} for all $k\ge n/s$ and completes the proof. \end{proof}

Let ${\mathcal G}_{\bf A}=(V, E_{\bf A})$ be the graph to describe the network topology in which
there is an edge   $(l,l')\in E_{\bf A}$  if and only if
$a(i,j)\ne 0$.
Define the minimum out/in degree of the network graph $ \mathcal{G}_{\bf A} $ by
$ \delta_{\rm out}(\mathcal{G}_{\bf A}) =\min_{1\le i\le n} \#\{j, \ a(i,j)\ne 0\}$ and
$\delta_{\in}(\mathcal{G}_{\bf A}) =\min_{1\le i\le n} \#\{j, \ a(j, i)\ne 0\}$ respectively.
Next we consider  the  scenario of the network topology that
\begin{equation}\label{deltaG.con}
\delta({\mathcal G}_{\bf A}):=\min(  \delta_{\rm out}(\mathcal{G}_{\bf A}),  \delta_{\rm in}(\mathcal{G}_{\bf A}))>n/2.
\end{equation}
In the above scenario, for each active node there are at least $\delta({\mathcal G})$ non-straggler to receive information from
and
to send information to.


\begin{prop}\label{lemma_A_sde_simple_eig-struct-1}
If \eqref{deltaG.con} holds, then  the matrix $ {\bf A}_{\rm sde} $ used in the CoDGraD algorithm \eqref{mainalgorithm} has simple eigenvalue one and all other eigenvalues contained in the open unit complex disk centered at the origin.
\end{prop}

\begin{proof} Following the argument used in
the proof of Proposition \ref{lemma_A_sde_simple_eig}, it suffices to prove
  \begin{equation} \label{Tandonlemma.pfeq0-struct-1}
b_2(i,j)\ne 0 \ {\rm for \ all} \ 1\le i,j\le n.
    \end{equation}
    \noindent  Set
$C_1=\{l, b(i,l)\ne 0\}$ and $
C_2=\{l, b(l, j)\ne 0\}$.
By the assumption on the decoding matrix ${\bf A}$,
they contain at least $\delta({\mathcal G})>n/2$
elements contained in $C_1$ and $C_2$ respectively, and hence  there exists $l_0\in C_1\cap C_2$.
Hence
  $$
b_2(i,j)  = \sum_{l=1}^n b(i,l) b_k(l,j) \ge b(i, l_0) b(l_0, j)
>0.$$
  \noindent
This proves \eqref{Tandonlemma.pfeq0-struct-1} and completes the proof.
 \end{proof}

\section{Coordinated Distributed Coding}

\label{Network_Forming_Coding}

\subsection{Network Formation}

\subsubsection{Network Detection}

A node gets activated and decides to form a network coded CoDGraD implementation. We signify this node as the \emph{coordinator node}. The coordinator node sends a message containing a label identifying the CoDGraD implementation, the transmitting node (i.e., which is itself), the accumulated path of the message (i.e., currently the node itself), and its public key.
When a neighboring node receives that message it transmits a message containing the label of the received message, the current transmitting node (i.e., which is the current node), the accumulated path of message (i.e., appending previous message path with the current transmitting node) and an encrypted message block. This neighboring node also sends in its encrypted message block part; its symmetric key with its identifier both encrypted with the coordinator node public key.
This process of transmitting a message and receiving a message then retransmitting continues henceforth until one of the two halting criteria is met.

\begin{enumerate}
\item {If the to-be-transmitted message at a node contains in its accumulated path an edge which is traversed twice in the same direction, then this message will not be re-transmitted.}

\item {A desired time that takes into the consideration the size of the network to be designed and its desired performance is set. After that time, the coordinator node would have received messages related to its CoDGraD implementation, that contain the label, the transmitting nodes and the accumulated paths.}
\end{enumerate}

Every time criterion 2 is met, the coordinator node forms from the accumulated paths the adjacency matrix of the interacting nodes. And according to the designed coding scheme and the allowed number of stragglers, the coordinator node decides on the nodes that need to join the network thus satisfying the stragglers and data partitions' redundancy thresholds. Thus it forms the network adjacency matrix and its related coding schemes encoding and decoding matrices ${\bf A_{\rm sde}}$ and ${\bf B}$, i.e., the  decoding matrix used for the weighing matrix and the gradient coding matrix used for information privacy. The coordinator node will also decrypts all encrypted data in the encrypted message block of the message, thus retrieving the symmetric keys of each node of the network that were encrypted with its public key.

\subsubsection{Network Forming and Coding}

Following the  protocol described in the preceding subsection, the coordinator node sends a new message containing the label of the network, an encrypted information containing the weights relative to the node that will join the network (i.e., $ {\bf A}_{\rm sde} $ row), and also an encrypted information containing the coefficients of the coded gradients of the node to join the network (i.e., $ {\bf B} $ row for computing $ \nabla{g}_{i} $). Subsequently, each node will recover this encrypted information through privacy symmetric keys between the coordinator node and itself.

More precisely, the coordinator node sends in the encrypted message block of the message the row of $ {\bf A}_{\rm sde} $ identified with node $ i $, for all nodes $ i \in \mathcal{G} $, each encrypted with the respective symmetric key of node $ i $ that was decrypted in the previous step. It also sends the weight of $ \nabla{f}_{i} $ (i.e., $ {\bf B}_{ji} $) that will be used in coding node $ i $ primary gradient in its neighbor $ j $ and the symmetric key of node $ j $, both encrypted with the symmetric key of node $ i $ which was decrypted in the previous step. It performs this operation for all nodes $ i $ and their respective neighbors $ j \in \mathcal{N}_{i} $, accordingly. In this message there are also the transmitting node and accumulated path information as before.

This message might also contain the adjacency matrix of the designed network and the adjacency matrix of the larger network that also contains nodes that are not allowed to join the network at this time due to not meeting the straggler threshold. All this adjacency information is also encrypted in such a way that only nodes of the designed network can access this information and only the information related to their neighborhoods.

\begin{figure}[H]
\centering
\includegraphics[trim=0 0 0 0cm , clip, width=8cm]{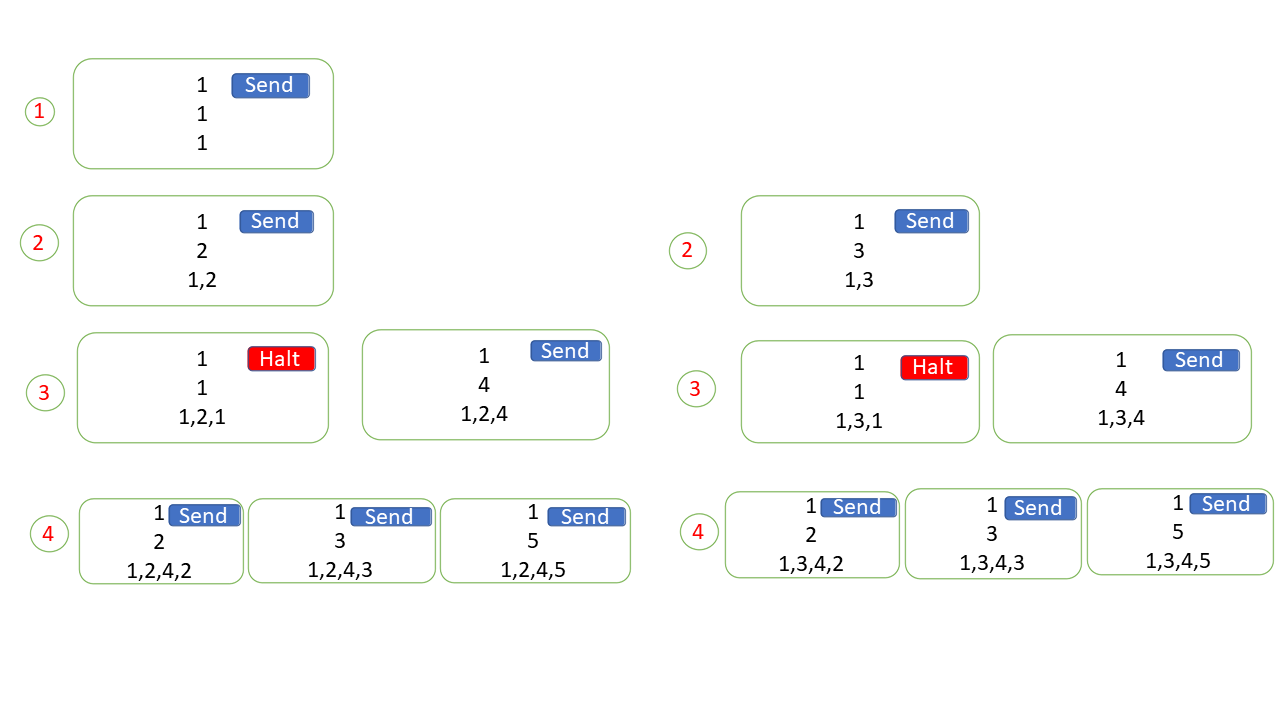}
\end{figure}

\vspace{-3cm}

\begin{figure}[H]
\centering
\includegraphics[trim=0 0 0 0cm , clip, width=8cm]{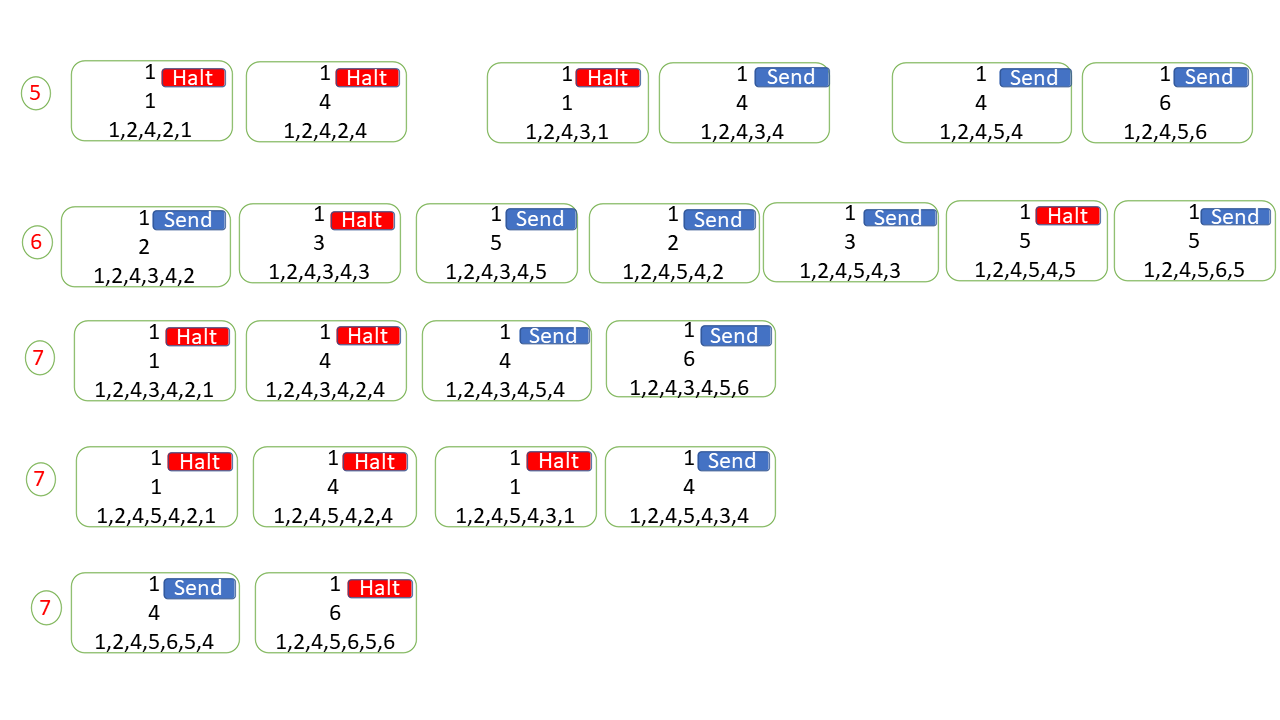}
\end{figure}

\vspace{-1cm}

\begin{figure}[H]
\centering
\includegraphics[trim=0 1cm 0 1cm , clip, width=8cm]{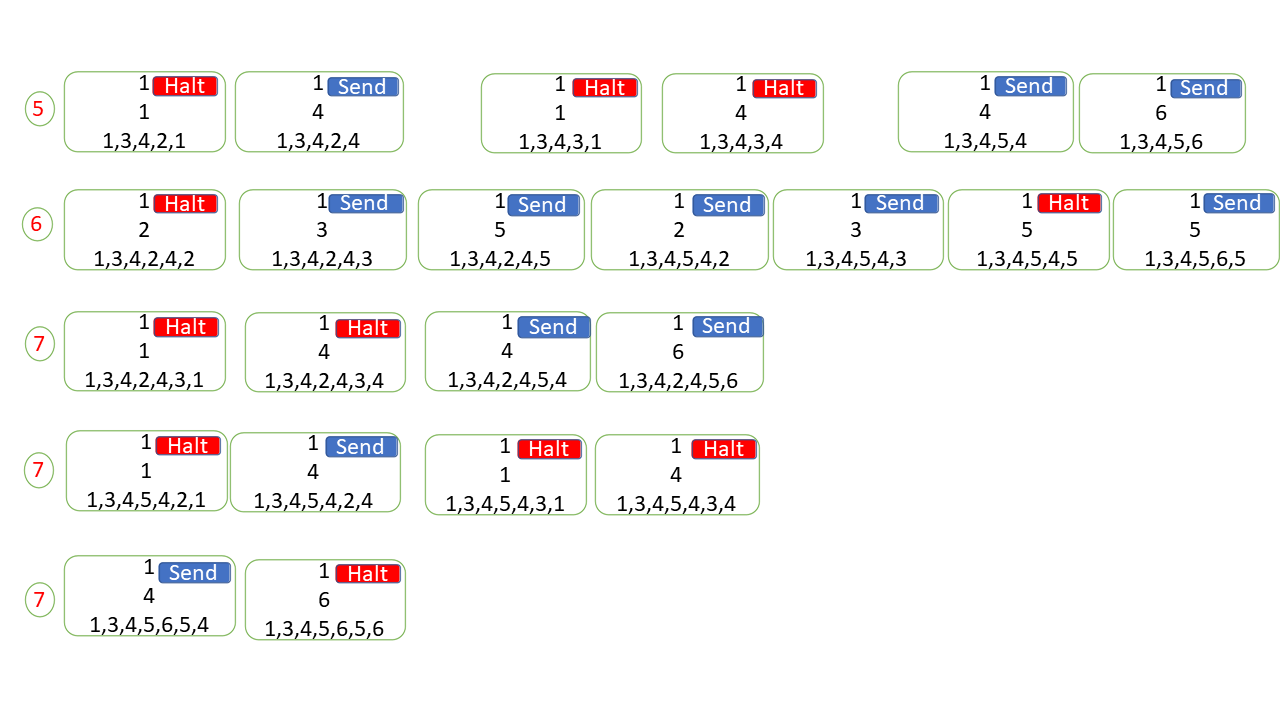}
\end{figure}

Meanwhile, the coordinator node and all other nodes of the network will separate their raw data into two parts. The primary information containing the initial raw data at the node and the secondary information containing the coded information formed from neighboring nodes due to CoDGraD coding.
When the coordinator node sends the previously described message it also sends its coded primary information to its neighboring vicinity. Thus, it sends its partition of raw data weighted with the weights $ {\bf B}_{j1} $ for all $ j \in \mathcal{N}_{1} $, each encrypted with the symmetric key of node $ j $ decrypted previously by the coordinator node. And obviously, the coordinator node (i.e., node $ 1 $) will have access of the row of $ {\bf A}_{\rm sde} $ found in the previous step.
As for the other nodes that are allowed by the coordinator node to join the network, when they receive this message they send the same message to their neighboring nodes with the new transmitting node information and the accumulated paths. They will also decrypt using their symmetric keys all the information related to them  in the encrypted message block. Thus, each node $ i $ will be able to decrypt the row of $ {\bf A}_{\rm sde} $ identified with it, and the weight $ {\bf B}_{ji} $ that will be used in coding the primary data partitions (i.e., the coded gradients $ \nabla{g}_{j} $) at node $ j $, together with the symmetric key of its neighbor $ j $, for all of its neighbors $ j \in \mathcal{N}_{i} $. 
At the first reception of this message the receiving node $ i $ also transmits in conjunction to the above described message (i.e., probably in a different channel) its primary raw information weighted by weights $ {\bf B}_{ji} $ for all $ j \in \mathcal{N}_{i} $, each encrypted with the symmetric key of node $ j \in \mathcal{N}_{i} $ decrypted previously. In the same way when a node $ j $ receives the encrypted primary raw data containing the data partitions used for evaluating the coded gradients $ \nabla{g}_{j} $ from each neighboring node $ i \in \mathcal{N}_{j} $, it decrypts with its symmetric key the coded data part of each of its neighbors $ i $ and stores it in its secondary information.


Thus, as it was mentioned before, by now each node will be able, according to the received encrypted coded primary information from neighboring nodes and the encrypted coefficients of the coding scheme, to form its weights matrix row part (i.e., $ {\bf A}_{\rm sde} $ row) and the coded secondary information needed in gradient coding (i.e., computes $ \nabla{g}_{i} $).

The message at a node is not transmitted again, if as before, it contains in its accumulated path an edge crossed in the same direction twice. Meanwhile, the primary information is transmitted only once or according to a designed protocol, one example would be, when all its neighbors send a message signifying that they received that information.
This process continues until the coordinator node receives all messages containing the sent encrypted information with the accumulated paths and conceives that all nodes in the desired network have coded their data and formed their weighting coefficients. Then it decides to implement the CoDGraD algorithm.

\begin{figure}[H]
\centering
\includegraphics[trim=0 0 0 0cm , clip, width=8cm]{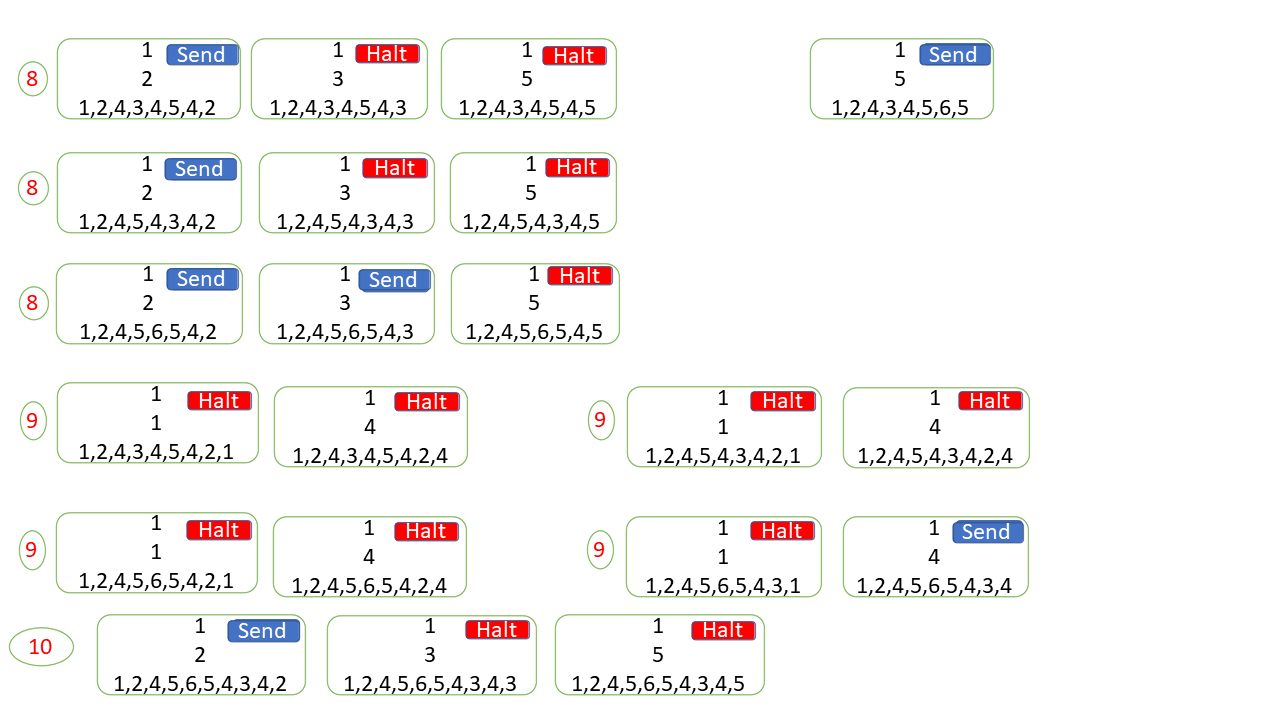}
\end{figure}

\vspace{-1.5cm}
\begin{figure}[H]
\centering
\includegraphics[trim=0 2 0 0cm , clip, width=8cm]{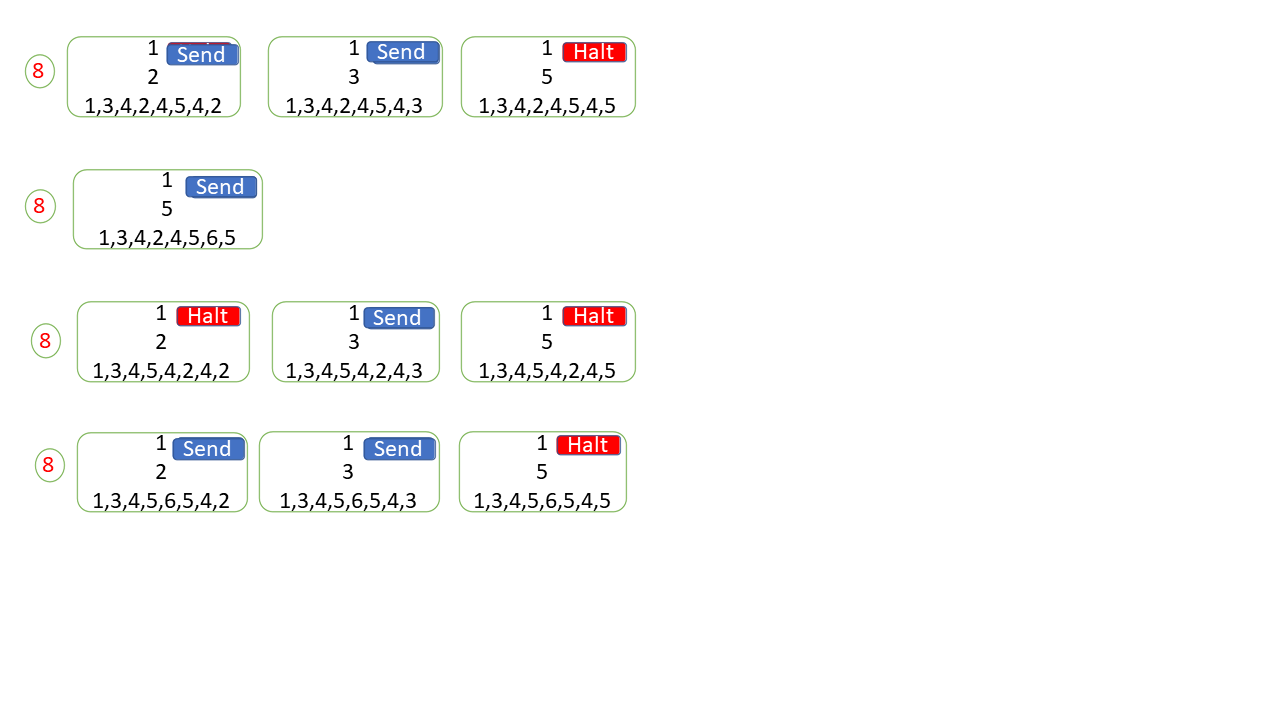}
\end{figure}

\vspace{-1.5cm}

\subsubsection{Adding nodes to already formed network}

While the CoDGraD algorithm is in process, when a node detects a message from a new active node it sends a encrypted message containing the updated neighborhood of this node with the new out of network node. This new node detection can be perceived directly through the new node sensing a source different from its usual neighbors or can also be recognized if the node detects a new neighbor not in the neighborhood adjacency matrix row which was communicated to it by the coordinator node in the previous step. Note that the latter policy is usually used if the update of the new adjacency matrix is accomplished only at the coordinator node and the first if any node in the network can perform such adjacency matrix upgrade. This follows for all nodes that detect new nodes. And when these messages are received by other nodes in the network they send this encrypted information to neighboring nodes only when they receive these messages for the first time. When the coordinator node receives these updated neighborhoods in a encrypted manner it deciphers the adjacency matrix of the resulting new network and if the new nodes meet the straggler threshold they are added to the network and a new updated network is formed. Henceforth, the coordinator node, as in stage 2, sends the encrypted messages of the adjacency matrices and coding schemes and allow the formation of a new network with new coding containing the allowed new nodes and thus implementing the CoDGraD on this new network.

\begin{figure}[H]
\centering
\includegraphics[trim=0 0 0 0cm , clip, width=8cm]{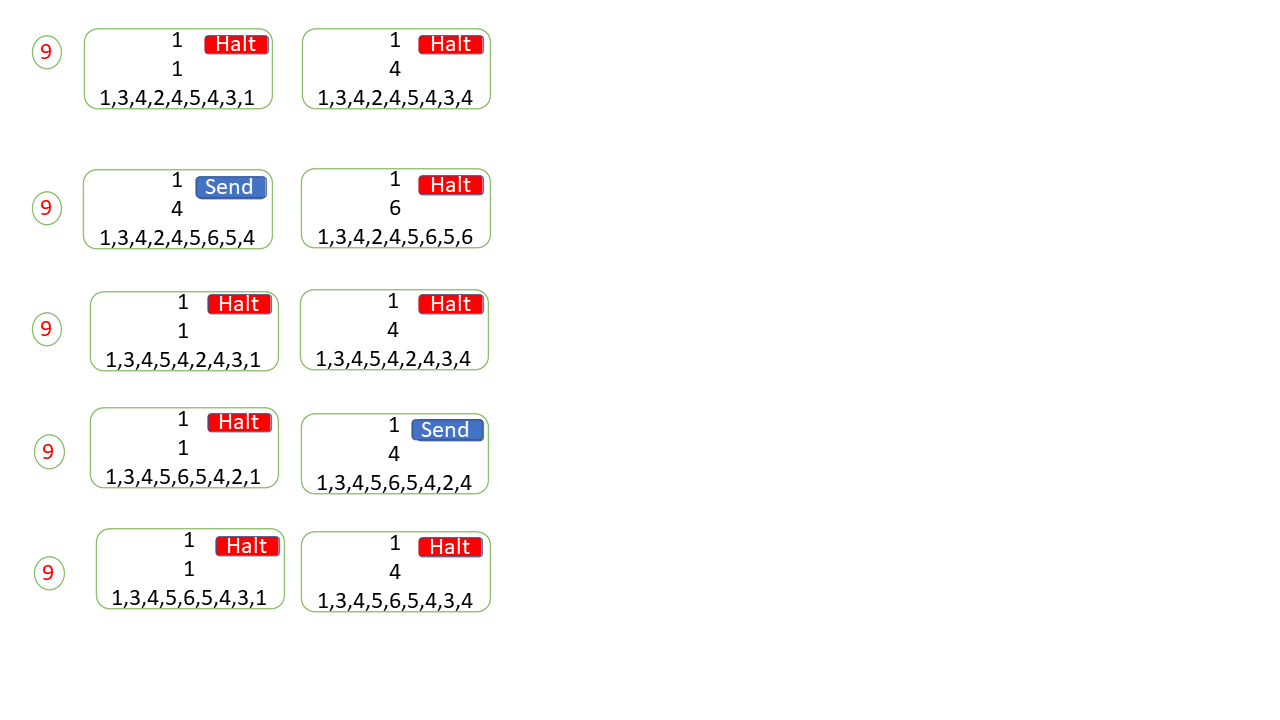}
\end{figure}

\vspace{-1cm}

\begin{figure}[H]
\centering
\includegraphics[trim=0 2 0 1cm , clip, width=8cm]{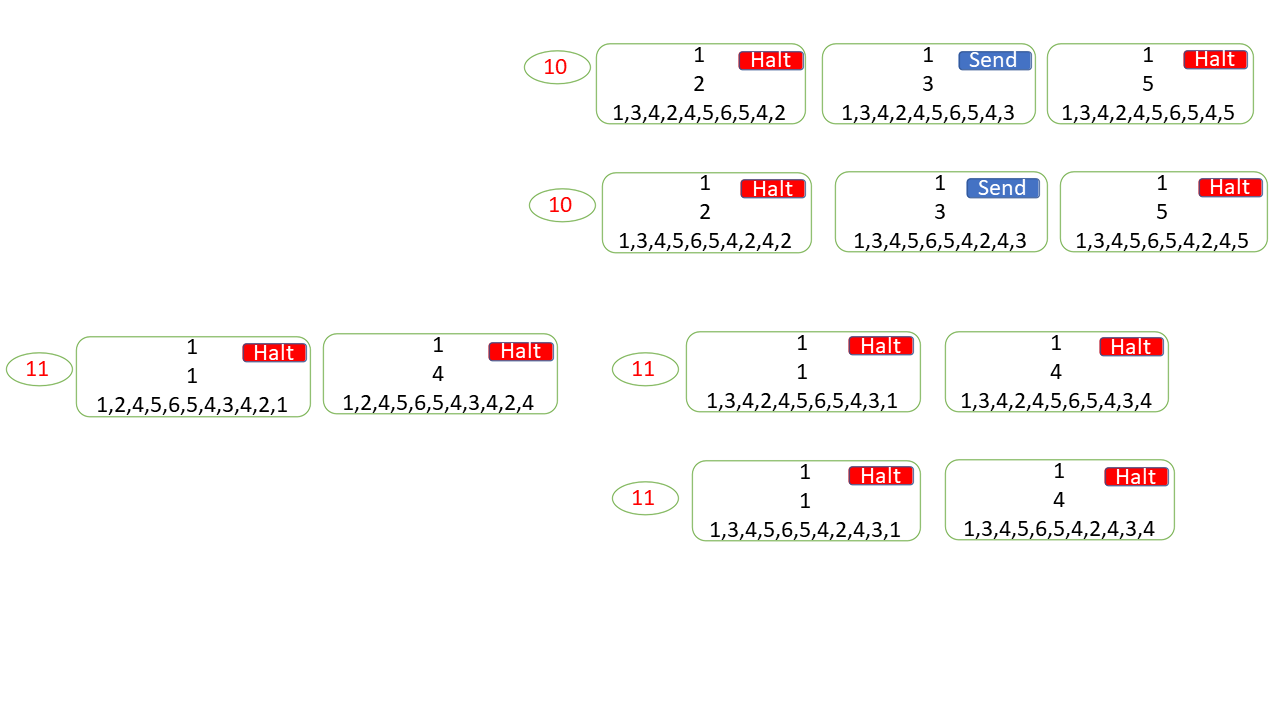}
\end{figure}

\vspace{-1cm}

It is worth mentioning that in our analysis we focused on static topology with fixed straggler nodes and thus fixed weighing and gradient coding matrices. Although we encrypted the information which allows upgrading to dynamic networks with privacy due to encrypting, however, we restrict our analysis to one weighing matrix and one coding scheme. Dynamic networks with the same nodes require protocol that uses the encoding matrix $ {\bf A} $ of all possible $ s+1 $-combinations, where $ s $ is the maximum number of allowed stragglers, and the process of adding new nodes needs the use of new weighing matrices (i.e., a new coding scheme), so we will leave their analysis to a future work. Although we can also use one coding scheme and one weighting matrix for a dynamic network, the one corresponding to the $ n - s $ non-straggler combinations but its performance will be considerably degraded.



Meanwhile, it is also worth noting that we could have allowed any node to upgrade the coding scheme whence it receives information about a neighborhood of a new node with the allowed straggler threshold. But we have restricted that to one node and specifically the coordinator node in order to preserve encrypted privacy keys designed due to that node. And thus not allowing the first approach since then we need to disclose the whole adjacency matrix to all nodes. By performing that we would be unable to preserve privacy through  allowing only neighborhood information to be disclosed to each node of the network while keeping the whole information to the coordinator node, the CoDGraD implementer. However, we can allow other nodes to send this information if they keep it encrypted through encrypting keys between them and the rest of the nodes.

\begin{figure}[H]\label{fig_network_6-node}
\centering
\includegraphics[trim=0 2 0 1cm , clip, width=8cm]{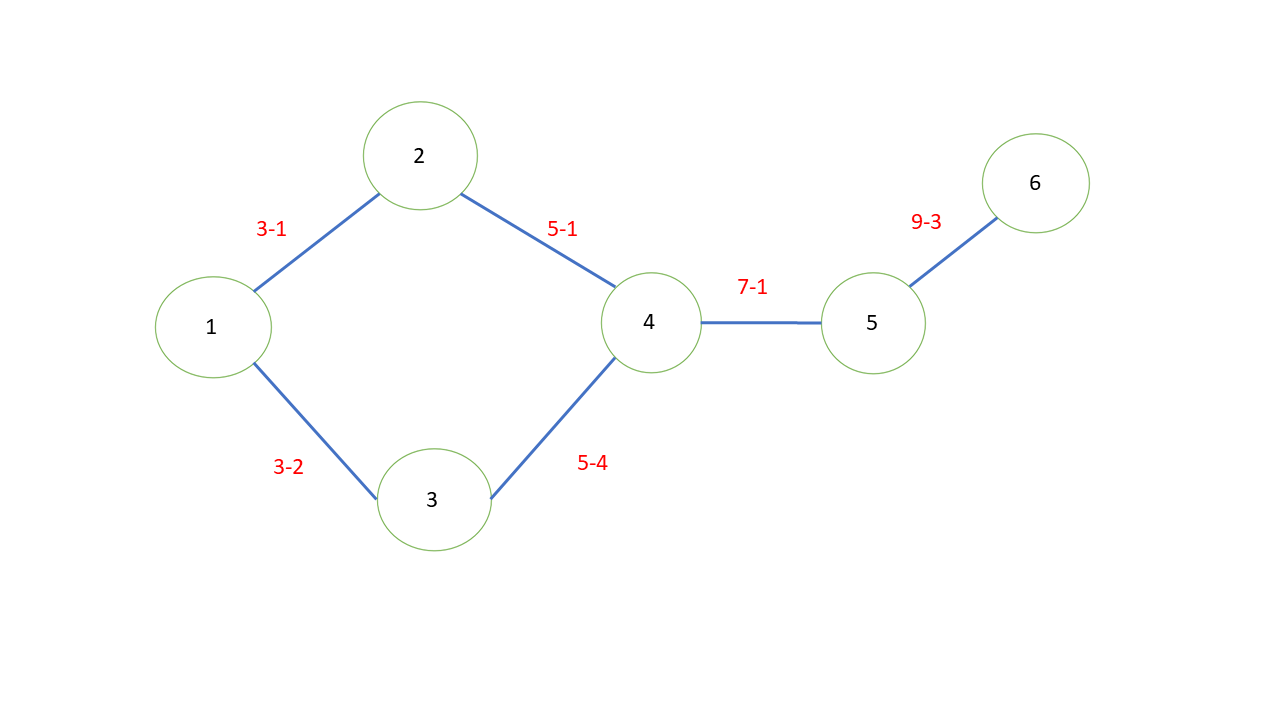}
\caption{Example of a 6-node network where coordinated distributed coded network forming is performed. The coefficients on the edges shows the corresponding message received by the coordinator node $ 1 $ that allows it to infer that edge connection. }
\end{figure}

\begin{remark}
We have provided in this section an example of the messaging protocol used in network forming on the $ 6 $-node network described in Fig~1. 
\end{remark}

\section{Consensus  property of the Code-based Distributed Gradient Descent Algorithm}\label{mainalg1.section}

In this  section, we consider the consensus property of  ${\bf x}_i(k),\  k \ge  1$, in the  CoDGraD  algorithm  \eqref{mainalgorithm}.


\begin{thm}\label{consensus.main.thm}
Let
${\bf x}_i(k),\  k\ge  1$,  be in the  CoDGraD  algorithm  \eqref{mainalgorithm}.
If the row stochastic  matrix
${\bf A}_{\rm sde}$ in \eqref{Afit.def} has simple eigenvalue one and all other eigenvalues contained in the open unit complex disk centered at the origin,
the sequence $\{\alpha_k\}_{k=0}^\infty$ satisfies
\eqref{alphak.assumption},
and  the local  objective functions $g_i, 1\le i\le n$,  have bounded gradients, i.e., there exists a positive constant $M$ such that
\begin{equation}\label{f.bounded}
\|\nabla g_i({\bf x})\|_2\le M,\  {\bf x}\in {\mathbb R}^N,
\end{equation}
then
\begin{equation}\label{consensus.newlimit}
\lim_{k\to \infty}(
{\bf x}_i(k)-{\bf x}_j(k))=0 \end{equation}
\noindent
and
\begin{equation}\label{consensus.newlimit2}
\lim_{k\to \infty}(
f({\bf x}_i(k))-f({\bf x}_j(k)))=0, \ 1\le i, j\le n. \end{equation}
\end{thm}

Let the row stochastic matrix ${\bf A}_{\rm sde}$ in \eqref{Afit.def} have simple eigenvalue one
and
$\lambda_m ({\bf A}_{\rm sde}), 1\le m\le 2n$,  be its  eigenvalues   listed in the order that
\begin{equation}\label{q.eigenvalues}
1=\lambda_1({\bf A}_{\rm sde})>|\lambda_2({\bf A}_{\rm sde})|\ge \ldots\ge  |\lambda_{2n} ({\bf A}_{\rm sde})|.\end{equation}
Write
${\bf A}_{\rm sde}= (q(i,j))_{1\le i, j\le 2n}$
and  for $1\le i\le n$, set
\begin{equation}\label{zi.def001}
{\bf z}_i(k)={\bf z}_{i+n}(k)={\bf x}_i(k),\  k\ge 0.\end{equation}
Then  the  CoDGraD  algorithm  \eqref{mainalgorithm} can be rewritten as
\begin{equation}\label{zk.def1}
{\bf z}_i(k+1)= \sum_{j=1}^{2n} q(i,j)\big({\bf z}_j(k)-\alpha_k {\bf h}_j(k)\big), \ 1\le i\le 2n,
\end{equation}
where
\begin{equation}\label{hi.def}
{\bf h}_i(k)= \left\{\begin{array} {ll} \nabla g_i({\bf z}_i(k))  & {\rm if} \ 1\le i\le n\\
-\nabla g_{i-n}({\bf z}_i(k)) & {\rm if} \ n+1\le i\le 2n.
\end{array}\right.
\end{equation}

Set
\begin{equation}\label{zh.def2}
 {\bf z}(k):=\big({\bf z}_i(k)\big)_{1\le i\le 2n}
\ \ {\rm and} \ \ {\bf h}(k):= \big({\bf h}_i(k)\big)_{1\le i\le 2n}
\end{equation}
  with vectors ${\bf z}_i(k)$
 and ${\bf h}_i(k)\in {\mathbb R}^N, 1\le i\le 2n$, as their $i$-th entries respectively.
Then   the iterative algorithm \eqref{zk.def1} can be reformulated in a matrix form:
\begin{equation}\label{zk.def2}
{\bf z}(k+1)= {\bf A}_{\rm sde}{\bf z}(k)-\alpha_k  {\bf A}_{\rm sde} {\bf h}(k), \ k\ge 0.
\end{equation}

Define
\begin{equation}\label{tildezk.def}
\tilde {\bf z}(k)=  {\bf z}(k) -{\bf P} {\bf z}(k), \ k\ge 0,\end{equation}
where
\begin{equation}\label{P0.def}
{\bf P}= {\bf 1}_{2n}  ({\bf a}_{\rm sde})^T
\end{equation}
\noindent
and
${\bf a}_{\rm sde}$ is  the  stationary probability vector  invariant under  the row stochastic matrix ${\bf A}_{\rm sde}$, i.e.,  the left eigenvector of ${\bf A}_{\rm sde}$ associated with eigenvalue one
that satisfies
\begin{equation}\label{asde.def}
{\bf a}_{\rm sde}^T {\bf A}_{\rm sde} = {\bf a}_{\rm sde}^T \  \ {\rm and} \  \  {\bf a}_{\rm sde}^T{\bf 1}_{2n}=1.
\end{equation}
\noindent
Then the consensus property \eqref{consensus.newlimit} reduces to establishing  
\begin{equation}\label{firstconvergence}
\lim_{k\to \infty} \|\tilde {\bf z}(k)\|_{2, \infty}=0,\end{equation}
where
$\|{\bf z}\|_{2, \infty}= \sup_{1\le i\le 2n} \|{\bf z}_i\|_2$
for a vector ${\bf z}= ({\bf z}_i)_{1\le i\le 2n}$ with entries ${\bf z}_i\in {\mathbb R}^N, 1\le i\le 2n$.

By    \eqref{P0.def} and \eqref{asde.def}, we have
\begin{equation}\label{QP.identity}
{\bf P} {\bf A}_{\rm sde}={\bf A}_{\rm sde}{\bf P}={\bf P}\ \ {\rm and}   \ \ {\bf P}^2= {\bf P}.
\end{equation}
This  together with
\eqref{zk.def2} implies that
\begin{equation}\label{tildezk.iteration}
\tilde {\bf z}(k+1) = ({\bf A}_{\rm sde}-{\bf P}) \tilde {\bf z}(k)-\alpha_k ({\bf A}_{\rm sde}-{\bf P}) {\bf h}(k).
\end{equation}
Applying  
 \eqref{tildezk.iteration} repeatedly yields
\begin{equation} \label{tildezk.eq0}
\tilde {\bf z}(k)   =  ({\bf A}_{\rm sde}-{\bf P})^k \tilde {\bf z}(0)-\sum_{l=0}^{k-1} \alpha_l
({\bf A}_{\rm sde}-{\bf P})^{k-l} {\bf h}(l), 
 \ k\ge 1.
\end{equation}
Therefore, we have the following estimate for $\|\tilde {\bf z}(k)\|_{2, \infty}, k\ge 1$ in Proposition 1, see
 Appendix \ref{tildezk.thm.appendix} for a detailed proof.

\begin{prop}\label{tildezk.prop}
Let ${\bf A}_{\rm sde}$, $\lambda_2({\bf A}_{\rm sde})$ and ${\bf P}$ be   as  in \eqref{Afit.def},  \eqref{q.eigenvalues}  and \eqref{P0.def} respectively.
Assume that  the row stochastic  matrix
${\bf A}_{\rm sde}$ has simple eigenvalue one and all other eigenvalues contained in the open unit complex disk centered at the origin,
and that the local  objective functions $g_i, 1\le i\le n$,  satisfy \eqref{f.bounded}.
Then there exists a positive constant $C_1$ such that
\begin{eqnarray}\label{tildezk.prop.eq1}
\|\tilde{\bf z}(k)\|_{2, \infty} & \hskip-0.08in \le &  \hskip-0.08in
    C_1 M
\sum_{l=0}^{k-1}  \Big(\frac{1+|\lambda_2({\bf A}_{\rm sde})|}{2}\Big)^{k-l} \alpha_l\nonumber\\
& &   \hskip-0.08in +
C_1 \Big(\frac{1+|\lambda_2({\bf A}_{\rm sde})|}{2}\Big)^k \|\tilde {\bf z}(0)\|_{2, \infty} \qquad
\end{eqnarray}
hold for all $k\ge 1$.
\end{prop}

\vskip-0.06in

By \eqref{P0.def},  we can write
\begin{equation}\label{xypm.def}
{\bf P}{\bf z}(k)=
\bar {\bf x}(k){\bf 1}_{2n} \ {\rm for \ some} \ \bar{\bf x}(k)\in {\mathbb R}^N.
\end{equation}
Observe that
$\|\tilde{\bf z}(k)\|_{2, \infty}= \max_{1\le i\le n} \|{\bf x}_i(k)-\bar {\bf x}(k)\|_{2}, \ k\ge 1$.
Then by Proposition \ref{tildezk.prop} we obtain the following estimate about the consensus property of ${\bf x}_i(k)$ for different $1\le i\le n$:  
\begin{eqnarray}\label{consensus.newthm.eq1}
& \hskip-0.08in  &  \hskip-0.08in \max_{1\le i\le n} \|{\bf x}_i(k)-\bar {\bf x}(k)\|_{2}
  \nonumber\\
  &\le & \hskip-0.08in  C_1 \Big(\frac{1+|\lambda_2({\bf A}_{\rm sde})|}{2}\Big)^k \max_{1\le i\le n} \|{\bf x}_i(0)-\bar {\bf x}(0)\|_{2}
  \nonumber\\
&  \hskip-0.08in & \hskip-0.08in    +  C_1 M
\sum_{l=0}^{k-1}  \Big(\frac{1+|\lambda_2({\bf A}_{\rm sde})|}{2}\Big)^{k-l} \alpha_l, \ \ k\ge 1.\qquad
\end{eqnarray}


Set $\gamma=\frac{1+|\lambda_2({\bf A}_{\rm sde})|}{2}\in (1/2, 1)$. Applying \eqref{consensus.newthm.eq1}
to our illustrative example \eqref{stepsize.example} of step sizes $\alpha_k= (k+1)^{-\theta}, k\ge 0$ for some $1/2<\theta\le 1$,
we can find a positive constant $C$ such that
\begin{eqnarray}\label{example.estimate1}
\hskip-0.08in & \hskip-0.08in & \hskip-0.08in  \max_{1\le i\le n} \|{\bf x}_i(k)-\bar {\bf x}(k)\|_{2}\nonumber\\
\hskip-0.08in &\hskip-0.08in\le & \hskip-0.08in  C_1  \max_{1\le i\le n} \|{\bf x}_i(0)-\bar {\bf x}(0)\|_{2}  \gamma^k\nonumber\\
\hskip-0.08in &\hskip-0.08in & \hskip-0.08in
     +  C_1 M
\sum_{l=0}^{k-1}  \gamma^{k-l} (k-l+a)^\theta (k+a)^{-\theta}\nonumber\\
\hskip-0.08in &\hskip-0.08in\le & \hskip-0.08in
 C (k+a)^{-\theta}, \ \ k\ge 0,
\end{eqnarray}
where the first inequality follows from \eqref{consensus.newthm.eq1}
and
the observation that  $a\ge 1$ and
$(k+a)^\theta\le  (l+a)^{\theta}(k-l+a)^{\theta}, 0\le l\le k$,
and the second estimate holds by the boundedness of the sequence $\gamma^k (k+a)^\theta, k\ge 0$, and
the convergence of the series $\sum_{m=0}^{\infty}  \gamma^{m} (m+a)^\theta$.

We finish this section with the proof of Theorem~\ref{consensus.main.thm}.

\begin{proof}[Proof of Theorem \ref{consensus.main.thm}] 
 By (IV.4) and (IV.17), we have
\begin{eqnarray}\label{tildezk.est001}
 & & \Big(\sum_{k=0}^\infty
\max_{1\le i\le n} \|{\bf x}_i(k)-\bar {\bf x}(k)\|_{2}^2\Big)^{1/2}\nonumber\\
 & \hskip-0.08in \le & \hskip-0.08in    C_1 M
 \Big(\sum_{k=0}^\infty\Big(\sum_{l=0}^{k-1}  \Big(\frac{1+|\lambda_2({\bf A}_{\rm sde})|}{2}\Big)^{k-l} \alpha_l\Big)^{2}\Big)^{1/2}
 \nonumber\\
 & \hskip-0.08in & \hskip-0.08in + C_1  \Big(\sum_{k=0}^\infty  \Big(\frac{1+|\lambda_2({\bf A}_{\rm sde})|}{2}\Big)^{2k}\Big)^{1/2}\nonumber\\
 & \hskip-0.08in & \hskip-0.08in \qquad
 \times \max_{1\le i\le n} \|{\bf x}_i(0)-\bar {\bf x}(0)\|_{2}\nonumber\\
   &  \hskip-0.08in \le & \hskip-0.08in \frac{  2 C_1  M  }  { 1-|\lambda_2({\bf A}_{\rm sde})| }
\Big(\sum_{k=0}^\infty \alpha_k^2\Big)^{1/2}\nonumber\\
&  \hskip-0.08in & \hskip-0.08in  +\frac{ 2 C_1}{\sqrt{1-|\lambda_2({\bf A}_{\rm sde})|}}  \max_{1\le i\le n} \|{\bf x}_i(0)-\bar {\bf x}(0)\|_{2},
\end{eqnarray}
where  the first inequality is obtained from (IV.17)
and
the triangle inequality for the space of square-summable sequence,
$$\left\{\Big(\frac{1+|\lambda_2({\bf A}_{\rm sde})|}{2}\Big)^{k}\right\}_{k=0}^\infty\ \ {\rm and} \ \
\left\{\sum_{l=0}^{k-1}  \Big(\frac{1+|\lambda_2({\bf A}_{\rm sde})|}{2}\Big)^{k-l} \alpha_l\right\}_{k=0}^\infty, $$
and the second estimate follows from
\begin{eqnarray*}
& \hskip-0.08in & \hskip-0.08in
\sum_{k=0}^\infty  \Big(\frac{1+|\lambda_2({\bf A}_{\rm sde})|}{2}\Big)^{2k}= \left(1-  \Big(\frac{1+|\lambda_2({\bf A}_{\rm sde})|}{2}\Big)^{2}\right)^{-1}\\
  & \hskip-0.08in \le & \hskip-0.08in \left(1-  \Big(\frac{1+|\lambda_2({\bf A}_{\rm sde})|}{2}\Big)\right)^{-1}
  \left(1+ \Big(\frac{1+|\lambda_2({\bf A}_{\rm sde})|}{2}\Big)^{2}\right)^{-1}\nonumber\\
  &  \hskip-0.08in \le  &   \hskip-0.08in \frac{4}{1- |\lambda_2({\bf A}_{\rm sde})|},
\end{eqnarray*}
and
\begin{eqnarray*}
 & \hskip-0.08in & \hskip-0.08in \sum_{k=0}^\infty\Big(\sum_{l=0}^{k-1}  \Big(\frac{1+|\lambda_2({\bf A}_{\rm sde})|}{2}\Big)^{k-l} \alpha_l\Big)^{2}\\
 & \hskip-0.08in \le & \hskip-0.08in
 \sum_{k=0}^\infty\Big(\sum_{l=0}^{k-1}  \Big(\frac{1+|\lambda_2({\bf A}_{\rm sde})|}{2}\Big)^{k-l} \alpha_l^2\Big)\\
 & \hskip-0.08in & \hskip-0.08in \times
 \Big(\sum_{l'=0}^{k-1}  \Big(\frac{1+|\lambda_2({\bf A}_{\rm sde})|}{2}\Big)^{k-l'} \Big)\\
 & \hskip-0.08in \le &  \hskip-0.08in \frac{2}{1-|\lambda_2({\bf A}_{\rm sde})|}  \sum_{k=0}^\infty\Big(\sum_{l=0}^{k-1}  \Big(\frac{1+|\lambda_2({\bf A}_{\rm sde})|}{2}\Big)^{k-l} \alpha_l^2\Big)\\
 & \hskip-0.08in \le &  \hskip-0.08in  \frac{2}{1-|\lambda_2({\bf A}_{\rm sde})|} \sum_{l=0}^\infty \sum_{k=l+1}^\infty \Big(\frac{1+|\lambda_2({\bf A}_{\rm sde})|}{2}\Big)^{k-l} \alpha_l^2\\
 & \hskip-0.08in \le & \hskip-0.08in  \left(\frac{2}{1-|\lambda_2({\bf A}_{\rm sde})|}\right)^2 \sum_{l=0}^\infty \alpha_l^2.
\end{eqnarray*}
Therefore  the limit in  \eqref{consensus.newlimit} follows as  the sequence
$\{\alpha_k\}_{k=0}^\infty$  is square summable by \eqref{alphak.assumption}.

 From \eqref{fgi.eq} and \eqref{weight.def} it follows that
$$\|\nabla f({\bf x})\|_2\le W^{-1}\sup_{1\le j\le n} \|\nabla g_j({\bf x})\|_2\le  M W^{-1},$$
where  $W=\max_{1\le i\le n} w_i$. This together with Proposition  \ref{tildezk.prop} implies that
\begin{eqnarray}\label{fconvergence.1}
\hskip-0.08in & &  |f({\bf x}_i(k))-f(\bar {\bf x}(k))|\nonumber\\
&\hskip-0.15in  \le & \hskip-0.13in \sup_{0\le t\le 1}
\|\nabla f( t{\bf x}_i(k)+(1-t)\bar {\bf x}(k))\|_2 \|{\bf x}_i(k)-\bar {\bf x}(k)\|_2\nonumber\\
  &\hskip-0.15in  \le & \hskip-0.13in    C_1 M^{2}  W^{-1}
\sum_{l=0}^{k-1}  \Big(\frac{1+|\lambda_2({\bf A}_{\rm sde})|}{2}\Big)^{k-l} \alpha_l\nonumber\\
& \hskip-0.15in & \hskip-0.13in    +
C_1 M W^{-1} \Big(\frac{1+|\lambda_2({\bf A}_{\rm sde})|}{2}\Big)^k  \max_{1\le i\le n} \|{\bf x}_i(0)-\bar {\bf x}(0)\|_{2}\nonumber\\
&  \hskip-0.15in  \to &   \hskip-0.13in  0 \ \ {\rm as} \  k\to \infty
\end{eqnarray}
for all $1\le i\le n$.
Then   the convergence
 \eqref{consensus.newlimit2} of the difference
$f({\bf x}_i(k)), k\ge 1$, between different $1\le  i\le n$
 follows. \end{proof}

\section{Convergence  property of the Code-based Distributed Gradient Descent Algorithm}\label{mainalg2.section}

In this  section, we consider  the convergence of  ${\bf x}_i(k),\  k\ge  1$, in the  CoDGraD  algorithm  \eqref{mainalgorithm},
\begin{equation}\label{consensus.limit}
\lim_{k\to \infty}
{\bf x}_i(k)=\bar {\bf x}, \ 1\le i\le n,\end{equation}
  where $\bar {\bf x}$ is the solution of the optimization problem \eqref{optimizationproblem}.
By \eqref{consensus.newlimit}, \eqref{xypm.def} and  \eqref{tildezk.est001},  it suffices to show that
$\bar {\bf x}(k), k\ge 1$, converges to $\bar {\bf x}$.

\begin{thm}\label{mainconvergence.thm}
  Assume that  the row stochastic  matrix ${\bf A}_{\rm sde}$
   in \eqref{Afit.def}
 has simple eigenvalue one and  all other eigenvalues contained in the open unit complex disk centered at the origin, the sequence $\{\alpha_k\}_{k=0}^\infty$ satisfies
\eqref{alphak.assumption},
 the local objective functions $g_j, 1\le j\le n$, satisfy  \eqref{f.bounded}  and
\vskip-0.16in
\begin{equation}\label{Lipschitz.con}
\|\nabla g_i(x)-\nabla g_i(y)\|_2\le L \|x-y\|_2,  \ x, y\in {\mathbb R}^N, 1\le i\le n,
\end{equation}
for some positive constant $L$, and the global objective function $f$  is strongly convex in the sense that
 \begin{equation}\label{accretive.con}
\langle \nabla f({\bf x})-\nabla f({\bf y}), {\bf x}-{\bf y}\rangle_N \ge A \|{\bf x}-{\bf y}\|_2^2
\end{equation}
for all ${\bf x}, {\bf y}\in B(\bar{\bf x}, C_2)$,  where
$\langle \cdot, \cdot\rangle_N$ is  the inner product on ${\mathbb R}^N$, $B(\bar{\bf x}, C_2)$ is the ball with center $\bar{\bf x}$ and radius $C_2$, and
 \begin{eqnarray*} C_2:& \hskip-0.08in = & \hskip-0.08in \exp\Big(\sum_{j=0}^\infty \alpha_j^2\Big)\Big\{  \|\bar{\bf x}(0)-\bar {\bf x}\|_2^2\nonumber\\
 & \hskip-0.08in  & \hskip-0.08in \quad  +
   \frac{ 8 C_1^2L^2 }{(1-|\lambda_2({\bf A}_{\rm sde}))|)^2 }
 \max_{1\le i\le n} \|{\bf x}_i(0)-\bar {\bf x}(0)\|_{2}^2
  \nonumber\\
 & \hskip-0.08in  & \hskip-0.08in \quad +   \frac{ M^2 (1+8 C_1^2) }{(1-|\lambda_2({\bf A}_{\rm sde})|)^2 }
  \Big(\sum_{j=0}^\infty \alpha_j^2\Big) 
  \Big\}.
 \end{eqnarray*}
  Then  $\bar{\bf x}(k), k\ge 1$, in \eqref{xypm.def} converges to the solution $\bar{\bf x}$ of
  the optimization problem \eqref{optimizationproblem},
\begin{equation}\label{barx.limit}
\lim_{k\to \infty} \bar {\bf x}(k)=\bar {\bf x}.
\end{equation}
\end{thm}

Combining  Theorems \ref{consensus.main.thm}
and \ref{mainconvergence.thm}, we have the following result on the convergence of the proposed CoDGraD algorithm.

\begin{thm}\label{overall.thm}
Let
the decoding matrix ${\bf A}_{\rm sde}$
and the objective function $f$ be as in Theorem \ref{mainconvergence.thm}, and step sizes
 $\{\alpha_k\}_{k=0}^\infty $ in the CoDGraD algorithm satisfy
\eqref{alphak.assumption}.
Then  the sequences $\{{\bf x}_i(k)\}_{k=1}^\infty, 1\le i\le n$
in the CoDGraD algorithm  converge to the optimal point $\bar {\bf x}$, i.e.,
$\lim_{k \rightarrow \infty}{\bf x}_{i}(k) = \bar{\bf x}, \ 1\le i\le n$.
\end{thm}

To prove Theorem \ref{mainconvergence.thm}, we need a technical lemma, which follows the probability  property for the vector ${\bf a}_{\rm sde}$ in \eqref{asde.def}. For the completeness of this paper, we include a detailed proof in Appendix \ref{weigh.appendix}.

\begin{lem}\label{weigh.lem}  Let ${\bf a}_{\rm sde}$ and $w_i, 1\le i\le n$, be as in \eqref{asde.def} and \eqref{weight.def}
respectively.
Set
\vskip-0.16in
\begin{equation}\label{weightvector.def}
{\bf w}=(w_1, \ldots, w_n, w_1, \ldots, w_n)^T
\ \ {\rm and} \ \
 \tilde w={\bf a}_{\rm sde}^T    {\bf w}.
 \end{equation}
 \vskip-0.06in
 \noindent  Then
 \begin{equation}\label{tildew.estimate}
 0<\min_{1\le i\le n} w_i\le \tilde w\le  \max_{1\le i\le n} w_i.
 \end{equation}
\end{lem}

By \eqref{P0.def} and \eqref{xypm.def}, we have
\begin{equation}\label{convergence.thm.pf.eq1}
\bar {\bf x}(k) = 
{\bf a}_{\rm sde}^T {\bf z}(k).
\end{equation}
This together with  \eqref{zk.def2}  leads to the following iterative algorithm for $\bar {\bf x}(k), k\ge 0$:
\begin{eqnarray}  \label{convergence.thm.pf.eq2}
\bar {\bf x}(k+1)   =  \bar {\bf x}(k) -\alpha_k {\bf a}_{\rm sde}^T   {\bf h}(k), \ k\ge 0.
\end{eqnarray}
For local objective functions $g_i, 1\le i\le n$,  with Lipschitz gradients  \eqref{Lipschitz.con},  we observe that
 ${\bf a}_{\rm sde}^T   {\bf h}(k)$ is an inexact estimate of the  scaled global gradient $\tilde w\nabla f(\bar{\bf x}(k))$, see  Appendix \ref{inexactgradient.thm.appendix} for a detailed proof.

\begin{prop}\label{inexactgradient.prop}
Let ${\bf A}_{\rm sde}$, $\lambda_2({\bf A}_{\rm sde})$, ${\bf P}$ and $\tilde w$ be   as  in \eqref{Afit.def},  \eqref{q.eigenvalues},
 \eqref{P0.def} and \eqref {weightvector.def} respectively.
Assume that the row stochastic  matrix
${\bf A}_{\rm sde}$ has simple eigenvalue one and all other eigenvalues contained in the open unit disk centered at the origin,
and the local objective functions $g_j,1\le j\le n$, satisfy \eqref{Lipschitz.con}.  Then
\begin{equation} \label{reminder.estimate}
\|{\bf a}_{\rm sde}^T   {\bf h}(k)- \tilde w \nabla f(\bar{\bf x}(k))\|_2
  \le   L \|{\bf a}_{\rm sde}\|_1  \|\tilde {\bf z}(k) \|_{2, \infty}.
\end{equation}
\end{prop}

By Proposition \ref{inexactgradient.prop}, the iterative algorithm \eqref{convergence.thm.pf.eq2}
can be considered as  the  gradient descent algorithm
\eqref{classicalgradient} with inexact gradient update.
Then following a standard argument, we have the boundedness of
 $\bar {\bf x}(k), k\ge 1$,  when the objective function $f$ is convex, see Appendix \ref{preconvergence.thm.appendix} for a detailed proof.

 \begin{prop}\label{preconvergence.prop}
 Let  the matrix ${\bf A}_{\rm sde}$,   the sequence $\{\alpha_k\}_{k=0}^\infty$
 and the local objective functions $g_j, 1\le j\le n$, be as in Theorem \ref{mainconvergence.thm}.
 If the global objective function $f$ is convex, then  
 \begin{equation}\label{preconvergence.thm.eq1}\|\bar {\bf x}(k)-\bar {\bf x}\|_2^{2} \le C_2 \ \ {\rm for \ all} \ \  k\ge 0,
 \end{equation}
 where  $C_2$ is the constant in Theorem \ref{mainconvergence.thm}.
 \end{prop}

The estimate in
\eqref {preconvergence.thm.eq1}  can be improved if the objective function $f$ has  the  strongly convex  property \eqref{accretive.con}, see Appendix \ref{preconvergence.thm.appendix} for a detailed proof.

\begin{prop}\label{convergence.prop}
 Let  the matrix ${\bf A}_{\rm sde}$,  the sequence $\{\alpha_k\}_{k=0}^\infty$,
 the local objective functions $g_j, 1\le j\le n$,  and the global objective function $f$ be as in Theorem \ref{mainconvergence.thm}.
  Then there exists a positive constant $C$ such that
 \begin{eqnarray}\label{convergence.thm.eq01}
 \|\bar{\bf x}(k)-\bar{\bf x}\|_2^2
 & \hskip-0.08in \le & \hskip-0.08in  C  \exp\Big(-\tilde w  A\sum_{j=0}^{k-1} \alpha_j\Big)
 \Big\{ \|\bar{\bf x}(0)-\bar {\bf x}\|_2^2  \nonumber\\
& \hskip-0.08in  & \hskip-0.08in  + \sum_{j=0}^{k-1} \exp\Big(\tilde w  A\sum_{l=0}^{j} \alpha_l\Big) \Big( M^2\alpha_j^2\nonumber \\
& & \hskip-0.08in + L^2 \max_{1\le i\le n} \|{\bf x}_i(j)-\bar {\bf x}(j)\|_2^2\Big)\Big\},
\end{eqnarray}
 hold for all $k\ge 2$.
\end{prop}

Applying \eqref{convergence.thm.eq01} for the illustrative examples \eqref{stepsize.example} of step sizes, and using \eqref{example.estimate1}, we can find a positive constant $\tilde C$ such that
\begin{eqnarray}\label{example.estimate2}
 & \hskip-0.08in  & \hskip-0.08in \|\bar{\bf x}(k)-\bar{\bf x}\|_2^2
  \le   \exp\Big(-\tilde w  A\sum_{l=0}^{k-1} (l+a)^{-\theta}\Big)
 \Big\{ \|\bar{\bf x}(0)-\bar {\bf x}\|_2^2  \nonumber\\
& \hskip-0.08in  & \hskip-0.08in  + (M^2+L^2 C^2)
 \sum_{j=0}^{k-1} \exp\Big(\tilde w  A\sum_{l=0}^{j} (l+a)^{-\theta}\Big) (j+a)^{-2\theta}\Big\}\nonumber \\
& \hskip-0.08in \le
& \hskip-0.08in
\exp\Big(-\frac{\tilde w  A}{1-\theta} \big((k+a)^{1-\theta}-a^{1-\theta}\big)\Big)
 \|\bar{\bf x}(0)-\bar {\bf x}\|_2^2\nonumber\\
& \hskip-0.08in
& \hskip-0.08in + (M^2+L^2 C^2) \exp\Big(-\frac{\tilde w  A}{1-\theta} (k+a)^{1-\theta}\Big)\nonumber\\
& \hskip-0.08in
& \hskip-0.08in \quad \times
\sum_{j=0}^{k-1} \exp\Big(\frac{\tilde w  A}{1-\theta} (j+1+a)^{1-\theta}\Big)
(j+a)^{-2\theta}\nonumber\\
 & \hskip-0.08in \le
& \hskip-0.08in
 \tilde C^2 (k+a)^{-\theta}, \  \ k\ge 1,\end{eqnarray}
where the first estimate holds by
\eqref{convergence.thm.eq01}, the second one follows from the observation
\begin{eqnarray*}
\sum_{l=m}^{k-1} (l+a)^{-\theta}  
& \le  &
 \frac{ (k+a)^{1-\theta}-(m+a)^{1-\theta}}{1-\theta}\end{eqnarray*}
 for $0\le m\le k-1$,
and the third one is obtained from
the boundedness of the sequences
$(k+a)^{\theta} \exp(-\frac{\tilde w  A}{1-\theta} (k+a)^{1-\theta})$ and $(k+a+1)^{1-\theta}-(k+a)^{1-\theta}, k\ge 0$, and
the integrability of $\int_0^\infty
\exp(-\frac{\tilde w  A}{1-\theta}  t) (t+1)^{\theta/(1-\theta)} dt$,
and the following estimate
\begin{eqnarray*}
 & \hskip-0.08in
& \hskip-0.08in \sum_{j=0}^{k-1} \exp\Big(\frac{\tilde w  A}{1-\theta} (j+1+a)^{1-\theta}\Big)
(j+a)^{-2\theta}\nonumber\\
 & \hskip-0.08in \le
& \hskip-0.08in  3^{2\theta}\int_0^k  \exp\Big(\frac{\tilde w  A}{1-\theta} (x+a+1)^{1-\theta}\Big) (x+a+1)^{-2\theta}dx
\\
 & \hskip-0.08in \le
& \hskip-0.08in  3^{2\theta} (1-\theta)^{-1}
\int_0^{a_k}
\exp\Big(\frac{\tilde w  A}{1-\theta}  \big((k+a+1)^{1-\theta}-t\big)\Big)\nonumber\\
 & \hskip-0.08in
& \hskip-0.08in \times
\big((k+a+1)^{1-\theta}-t\big)^{-\theta/(1-\theta)} dt\\
 & \hskip-0.08in
\le & \hskip-0.08in
3^{2\theta} (1-\theta)^{-1}
\exp\Big(\frac{\tilde w  A}{1-\theta}  (k+a+1)^{1-\theta}\Big)
(k+a+1)^{-\theta}
\nonumber\\
 & \hskip-0.08in
& \hskip-0.08in \times
\int_0^{a_k}
\exp\Big(-\frac{\tilde w  A}{1-\theta}  t\Big) (t+1)^{\theta/(1-\theta)} dt\\
 & \hskip-0.08in
\le & \hskip-0.08in
3^{2\theta} (1-\theta)^{-1}
\exp\Big(\frac{\tilde w  A}{1-\theta}  (k+a+1)^{1-\theta}\Big)
(k+a)^{-\theta}
\nonumber\\
 & \hskip-0.08in
& \hskip-0.08in \times
\int_0^\infty
\exp\Big(-\frac{\tilde w  A}{1-\theta}  t\Big) (t+1)^{\theta/(1-\theta)} dt,
\end{eqnarray*}
where $a_k=(k+1+a)^{1-\theta}-(a+1)^{1-\theta}\le (k+1+a)^{1-\theta}-1$.
Therefore
\begin{equation}\label{example.estimate2}\|\bar{\bf x}(k)-\bar {\bf x}\|\le \tilde C (k+a)^{-\theta/2},\ k\ge 1.\end{equation}
This implies that $\bar{\bf x}(k), k\ge 0$, has faster convergence to the limit $\bar{\bf x}$ when we select larger $\theta\in (1/2, 1)$,
however comparing
 the convergence rate $ O( \log k/\sqrt{k}) $ for the uncoded incremental gradient methods \citep{nedic2001,gurbuzbalaban2015},
 the convergence rate in \eqref{example.estimate2} is
always slow, even it is very close when $\theta$ is close to one.
After careful verification in the above estimation,  we observe that
the constant $\tilde C$ could be smaller if the second eigenvalue $\lambda_2({\bf A}_{\rm sde})$ associated with the decoding matrix ${\bf A}$ is smaller.

Set  $W=\max_{1\le i\le n} w_i$.
 Observe that
 $$ f(\bar {\bf x}(k))- f(\bar{\bf x})= \langle \nabla f (t \bar {\bf x}(k))+(1-t)\bar{\bf x}), \bar {\bf x}(k)- \bar{\bf x}\rangle$$
for some $0\le t\le 1$, and that
$$ \|\nabla f({\bf x})\|_2 \le   \|\nabla f({\bf x})- \nabla f(\bar {\bf x})\|_2
\le L  W^{-1} \|{\bf x}-\bar {\bf x}\|_2,
$$
where the second inequality follows from \eqref{Lipschitz.con} and the row stochastic property for the matrix ${\bf A}_{\rm sde}$.
This together with  Proposition \ref{convergence.prop} proves that 
  \begin{eqnarray}\label{convergence.cor.eq01}
f(\bar{\bf x})& \hskip-0.08in \le & \hskip-0.08in  f(\bar{\bf x}(k))\le
   f(\bar{\bf x})+L  W^{-1} \|{\bf x}-\bar {\bf x}\|_2\nonumber\\
&\hskip-0.08in \le  & \hskip-0.08in    f(\bar{\bf x})+L  W^{-1}  \exp\Big(- \tilde w A\sum_{j=0}^{k-1} \alpha_j\Big)
  \nonumber\\
& \hskip-0.08in  & \hskip-0.08in \times \Big\{ \|\bar{\bf x}(0)-\bar {\bf x}\|_2^2  + \sum_{j=0}^{k-1} \exp\Big(\tilde w  A\sum_{l=0}^{j} \alpha_l\Big) \nonumber \\
& & \hskip-0.08in \times \Big( M^2\alpha_j^2 + L^2 \max_{1\le i\le n} \|{\bf x}_i(j)-\bar {\bf x}(j)\|_2^2\Big)
\Big\} 
\qquad 
\end{eqnarray}
 hold for all $k\ge 2$.
For the case that step size $\alpha_k, k\ge 0$ chosen as in \eqref{stepsize.example}, we can use the similar argument used to prove
\eqref{example.estimate2}
and show that
  \begin{eqnarray}\label{convergencerate.example}
  |f(\bar{\bf x}(k))-f(\bar{\bf x})| \le
  C (k+a)^{-\theta},\   k\ge 1,
\end{eqnarray}
for some positive constant $C$.

We finish this section with the proof of Theorem \ref{mainconvergence.thm} under the assumption that Proposition \ref{convergence.prop} holds.

\begin{proof}[Proof of Theorem \ref{mainconvergence.thm}]
By  \eqref{alphak.assumption} and \eqref{tildezk.est001}, we have
 \begin{equation}\label {mainconvergence.thm.pf.eq11}
 \sum_{j=0}^{\infty}  M^2\alpha_j^2 + L^2 \max_{1\le i\le n} \|{\bf x}_i(j)-\bar {\bf x}(j)\|_2^2<\infty,
 \end{equation}
 and
 \begin{equation} \label{mainconvergence.thm.pf.eq10}
\lim_{k\to \infty} \exp \Big(-\tilde w  A  \sum_{j=l}^
k \alpha_j\Big)=0, \ l\ge 0.\end{equation}
Combining \eqref{convergence.thm.eq01},
\eqref{mainconvergence.thm.pf.eq11} and \eqref{mainconvergence.thm.pf.eq10}
proves the desired limit  in \eqref{barx.limit} by the dominated convergence theorem.
\end{proof}

\section{Numerical Simulations}
\label{simulation.section}

In this section, we consider  the following unconstrained convex optimization problem on a network  
\begin{equation}\label{objective.simulation}
 \arg\min_{{\bf x}\in \mathbb{R}^{N}} \|  {\bf G}{\bf x}- {\bf y}  \|_{2}^{2},\end{equation}
where   the network contains  $n$ regions with  each region of the
partition equipped with a worker,  ${\bf G}$ is a random  matrix of size $Q\times N$ whose entries are independent  and identically distributed standard normal random variables,  and
\begin{equation} {\bf y}={\bf G} {\bf x}_o\in \mathbb{R}^Q\end{equation}
 has entries of ${\bf x}_o$ being identically independent random variables sampled from the uniform bounded random distribution  between $ -1 $ and $ 1 $.
 The solution $\bar {\bf x}$ of the above optimization problem is the least squares solution of the overdetermined system  
${\bf y}  = {\bf G}{\bf x}_o, {\bf x}_o\in \mathbb{R}^N$. 
In this section,  we   demonstrate the  performance of the CoDGraD  algorithm  \eqref{mainalgorithm} to solve the convex optimization problem
\eqref{objective.simulation} and also compare it with the performance of
 the conventional distributed gradient descent  algorithm (DGD)
under the CTA prototype \eqref{CTA}.

Assume that the network has $n$ active nodes. Then we can repartition  the network into $n$ regions around those $n$ nodes, and accordingly, the random measurement matrix ${\bf G}$,  the measurement data ${\bf y}$,  and the objective function  $f({\bf x}):= \|  {\bf G}{\bf x}- {\bf y}  \|_{2}^{2}$ in \eqref{objective.simulation} as follows,
 $$f({\bf x})= \sum_{i=1}^{m} f_i({\bf x}):= \sum_{i=1}^{n} \|{\bf G}_i {\bf x}-{\bf y}_i\|_2^2.  $$
 In our simulations, we assume that the repartitioned regions have the same size, i.e., the number of rows in ${\bf G}_i$
and lengths of vectors ${\bf y}_i$, for all $ i $ where $ 1\le i\le n$ are all equal.
Shown in Figure \ref{fig35_b} are two undirected graphs to describe
data exchanging structure for active nodes of a 3-node  and 5-node network respectively.
\begin{figure}[t]
\centering
\includegraphics[width=33mm, height=23mm]
{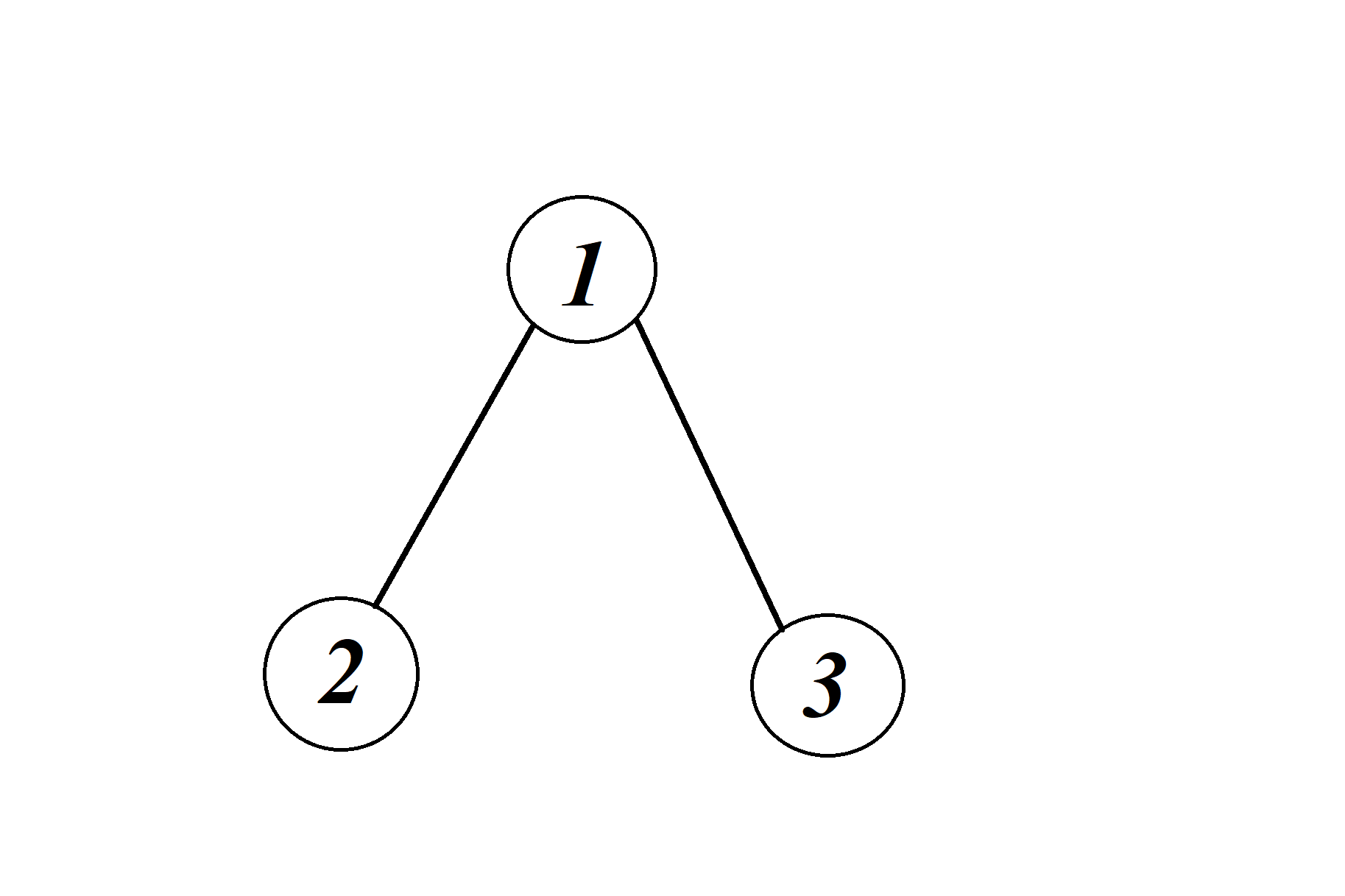}
\includegraphics [width=33mm, height=25mm]
{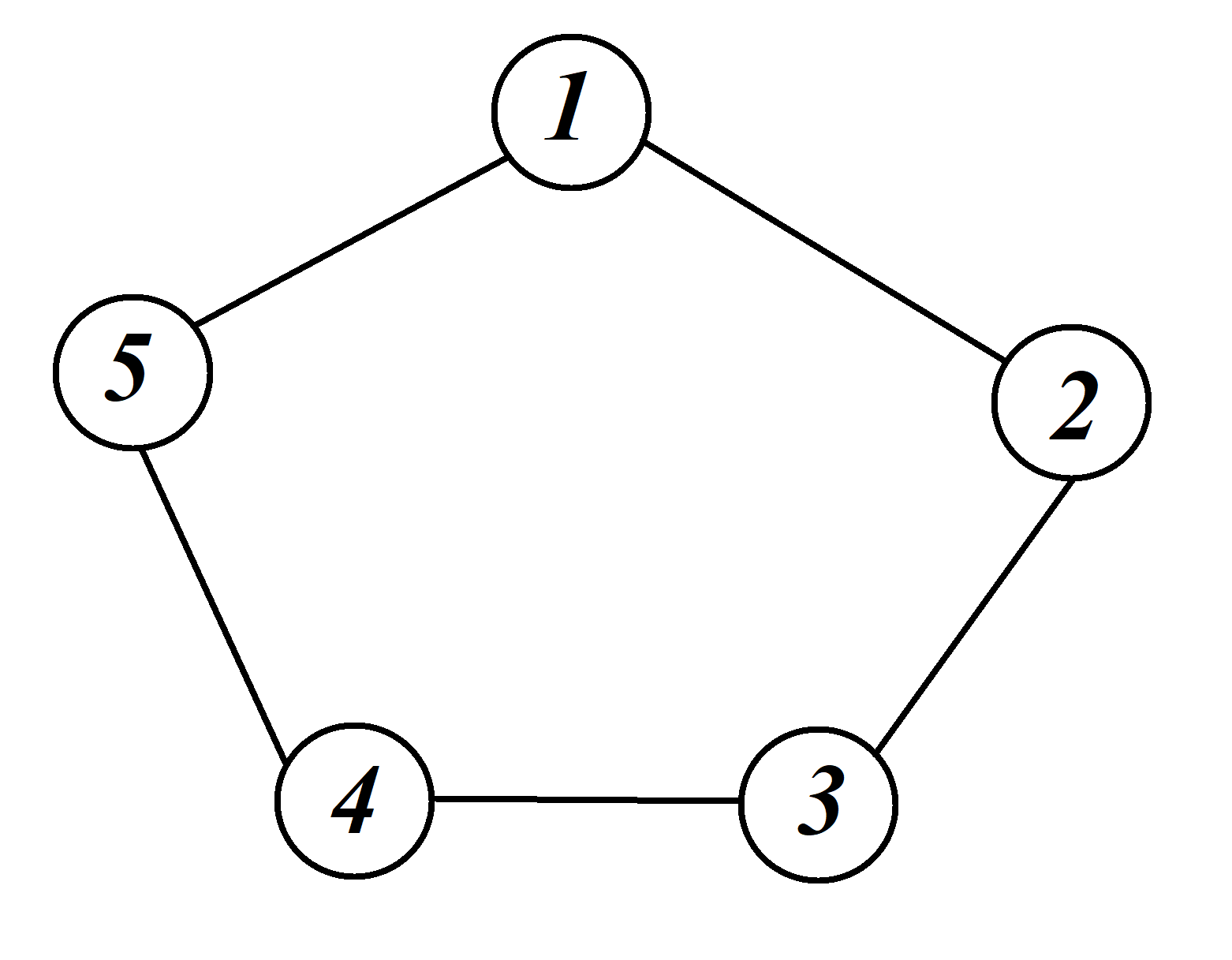}
\caption{Data exchanging structures with three/five-node network}
\label{fig35_b}
\end{figure}

In our simulations, we take $ (Q, N)= (225, 75) $ for Figure  
 \ref{fig_4} and $ (M, N)= (250, 50) $ for  Figure 
 \ref{fig_6}.
We use absolute error
$$ {\rm AE}:= \max_{1\le i\le n} {\|{\bf x}_i(k)-{\bf x}_o\|_2}/{\|{\bf x}_0\|_2} $$
and  consensus error
$${\rm CE}:= \max_{1\le i\le n} {\|{\bf x}_i(k)-\bar{\bf x}(k)\|_2}/{\|{\bf x}_o\|_2} $$
 to measure  the performance of the CoDGraD  algorithm
\eqref{mainalgorithm} and the DGD  algorithm \eqref{CTA}, where $n$ is the number of active nodes in the network.

In the first simulation where there are $ 3 $ nodes with its topology described on the left of Figure \ref{fig35_b}, we
take the coding matrix  ${\bf B}$ and the decoding matrix ${\bf A}$ as follows:
 \vskip-0.16in
 \begin{equation}\label{3_node_1_B}
     {\bf B}=   \left(\begin{array}{ccc}
        1 &	-\frac{5}{4} &	0 \\
0 &	1 & \frac{4}{9}\\
\frac{9}{5} &	0 & 	1
        \end{array}\right) \ {\rm and} \
              {\bf A} =  \left(\begin{array}{ccc}
        0 &	1	& \frac{5}{9} \\
1 &	\frac{9}{4} &	0 \\
-\frac{4}{5} &	0 &	1
        \end{array}\right).
        \end{equation}
 \noindent  The above coding/decoding matrix pair satisfies \eqref{AB1.eq} and
the  corresponding row stochastic matrix in \eqref{Afit.def}   is
 \begin{equation*}\label{3_node_1_A_sde}
       {\bf A}_{\rm sde} = \left(\begin{array}{cccccc}
       0 &	9/14	& 2/7 &	0 &	0 &	0 \\
4/13 &	9/13 &	0 &	0 &	0 &	0 \\
0 &	0 &	5/9 &  4/9 &	0 &	0 \\
0 &	9/14	& 5/14 &	0 &	0 &	0 \\
4/13 &	9/13 &	0 &	0 &	0 &	0 \\
0 &	0 &	5/9 &  4/9 &	0 &	0
        \end{array}\right).
            \end{equation*}
In Figure  
\ref{fig_4}, we 
present the  performance of  the CoDGraD  algorithm
\eqref{mainalgorithm} and the DGD  algorithm \eqref{CTA} with  absolute  and consensus metric being the average
of the corresponding metrics over 100 trials, where
random measurement matrix
${\bf G}$ has independent  and identically distributed standard normal random variables as its entries,
the original vector ${\bf x}_o$ is identically independent  random variables uniform distributed in $[-1, 1]$,
and
step sizes are $ \alpha_{k}=(k+300)^{-0.75} $ and $(k+500)^{-0.85}, k\ge 0 $, respectively.
\begin{figure}[t]
\centering
\includegraphics [width=78mm, height=45mm]{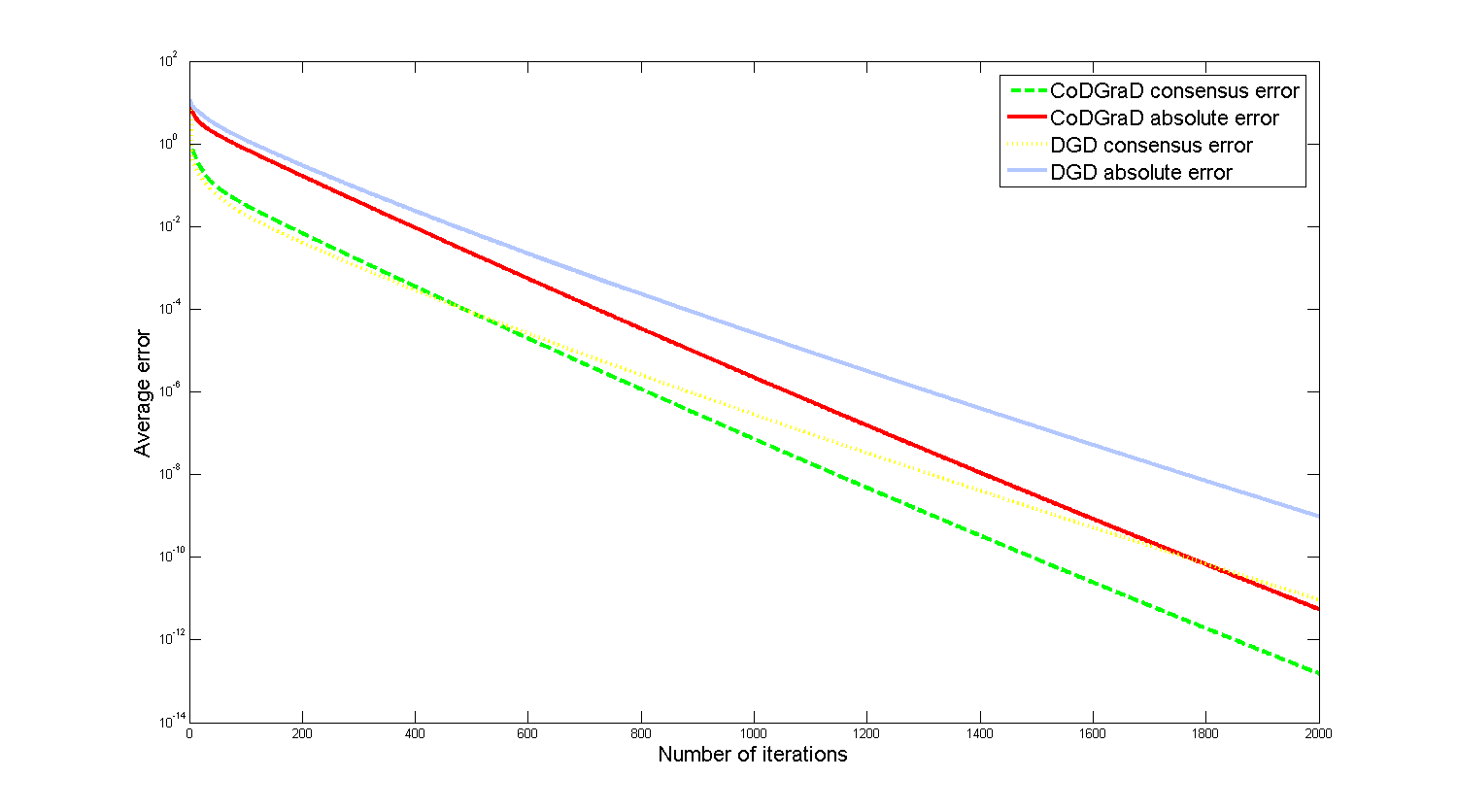}

\includegraphics  [width=78mm, height=45mm] 
{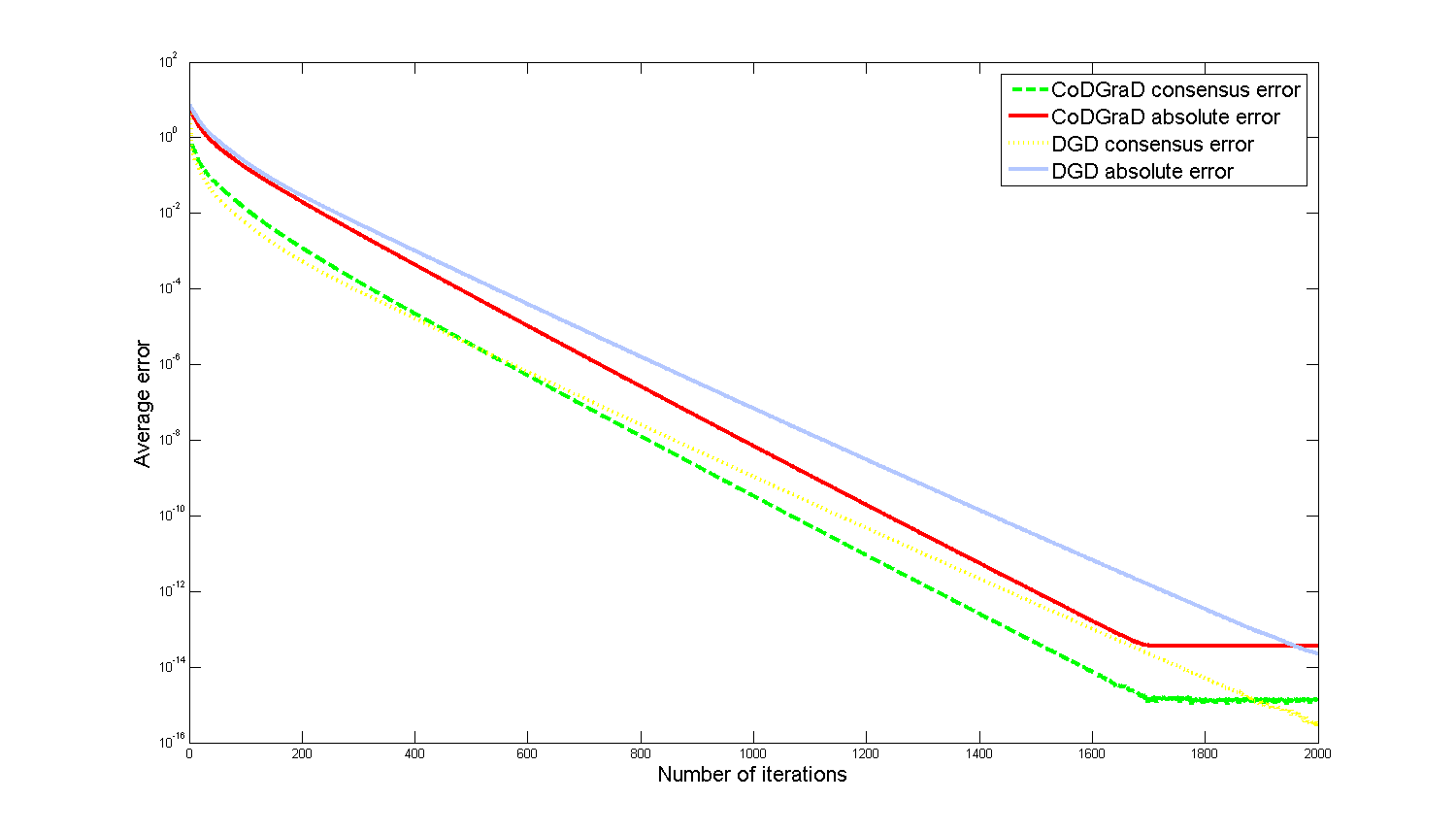}
\caption{Performance comparison of the CoDGraD  algorithm
\eqref{mainalgorithm} and the DGD algorithm \eqref{CTA}
over a three active nodes network  with $ (Q, N)= (225, 225) $ and step sizes $ \alpha_{k}= (k+300)^{-0.75}$ (top) and $(k+500)^{-0.85}, k\ge 0$ (bottom).  }
\label{fig_4}
\end{figure}

In the second simulation, the network has $ 5 $ active nodes with
data exchanging structure described on the right of Figure \ref{fig35_b}.
In that simulation,
the coding/decoding matrices are given by
 \vskip-0.16in
\begin{equation}\label{5_node_B}
    {\bf B}=   \left(\begin{array}{ccccc}
   1 & 2 & 1/2 & 0 & 0 \\
    0 & -1 & 3 &  4 & 0 \\
    0 &  0 & -5/2 & -3 & 1 \\
    1 & 0 &  0 & 1/5 & 13/5 \\
    2 & 1 & 0 & 0 & 4
        \end{array}\right)
\end{equation}
 \vskip-0.06in
 \noindent
 and
 \vskip-0.16in
 \begin{equation}\label{5_node_A}
      {\bf A} =  \left(\begin{array}{ccccc}
        1/2 & 1/4 & 0   &  0  & 1/4 \\
        1 &  1 & 1  &    0  &   0 \\
0  &   -1  &  -8/5 &  1 & 0  \\
0  &  0  & -2/5 & -1 & 1 \\
2 & 0  &  0  &   5  &  -3
        \end{array}\right)
\end{equation}
 \vskip-0.06in
 \noindent respectively.
The above coding/decoding matrix pair satisfies \eqref{AB1.eq} and
the  corresponding row stochastic matrix in \eqref{Afit.def}   is
\begin{equation*}\label{5_node_A_sde}
       {\bf A}_{\rm sde} = \left(\begin{array}{cccccccccc}
\frac{1}{2} & \frac{1}{4} & 0   &  0  & \frac{1}{4}  &    0  & 0  &   0  &   0  &    0 \\
 \frac{1}{3} &  \frac{1}{3} & \frac{1}{3}  &    0  &   0  &   0  & 0   &   0  &   0   & 0 \\
0  &   0  &  0  &  \frac{5}{18} & 0 & 0  & \frac{5}{18} & \frac{8}{18}   & 0   &   0 \\
0  &  0  & 0 & 0 & \frac{5}{12} & 0  & 0 & \frac{1}{6}  & \frac{5}{12}  &  0 \\
\frac{1}{5} & 0  &  0  &   \frac{1}{2}  &   0 &  0 &  0  &  0  &  0   &  \frac{3}{10}  \\
\frac{1}{2} & \frac{1}{4} & 0   &  0  & \frac{1}{4}  &    0  & 0  &   0  &   0  &    0 \\
\frac{1}{3} & \frac{1}{3}  & \frac{1}{3}   &    0  &   0  &   0  & 0   &   0  &   0   & 0 \\
0  &   0  &  0  &  \frac{5}{18} & 0 & 0  & \frac{5}{18} &  \frac{8}{18}  & 0   &   0 \\
0  &  0  & 0 & 0 & \frac{5}{12} & 0  & 0 & \frac{1}{6}  & \frac{5}{12}  &  0 \\
\frac{1}{5} & 0  &  0  &  \frac{1}{2}   &   0 &  0 &  0  &  0  &  0   & \frac{3}{10}
        \end{array}\right).
            \end{equation*}
Shown in Figure  
\ref{fig_6} is the performance of the CoDGraD  algorithm \eqref{mainalgorithm} and the DGD algorithm \eqref{CTA}, where the absolute metric and consensus metric are the average
of the corresponding metrics over 100 trials with random measurement matrix
${\bf G}$    and
the original vector ${\bf x}_o$ being  selected as in the first simulation,
and
step sizes being $ \alpha_{k}=(k+800)^{-0.90} $ and $(k+500)^{-0.95}, k\ge 0 $, respectively.

\begin{figure}[t]
\centering
\includegraphics  [width=78mm, height=45mm] 
{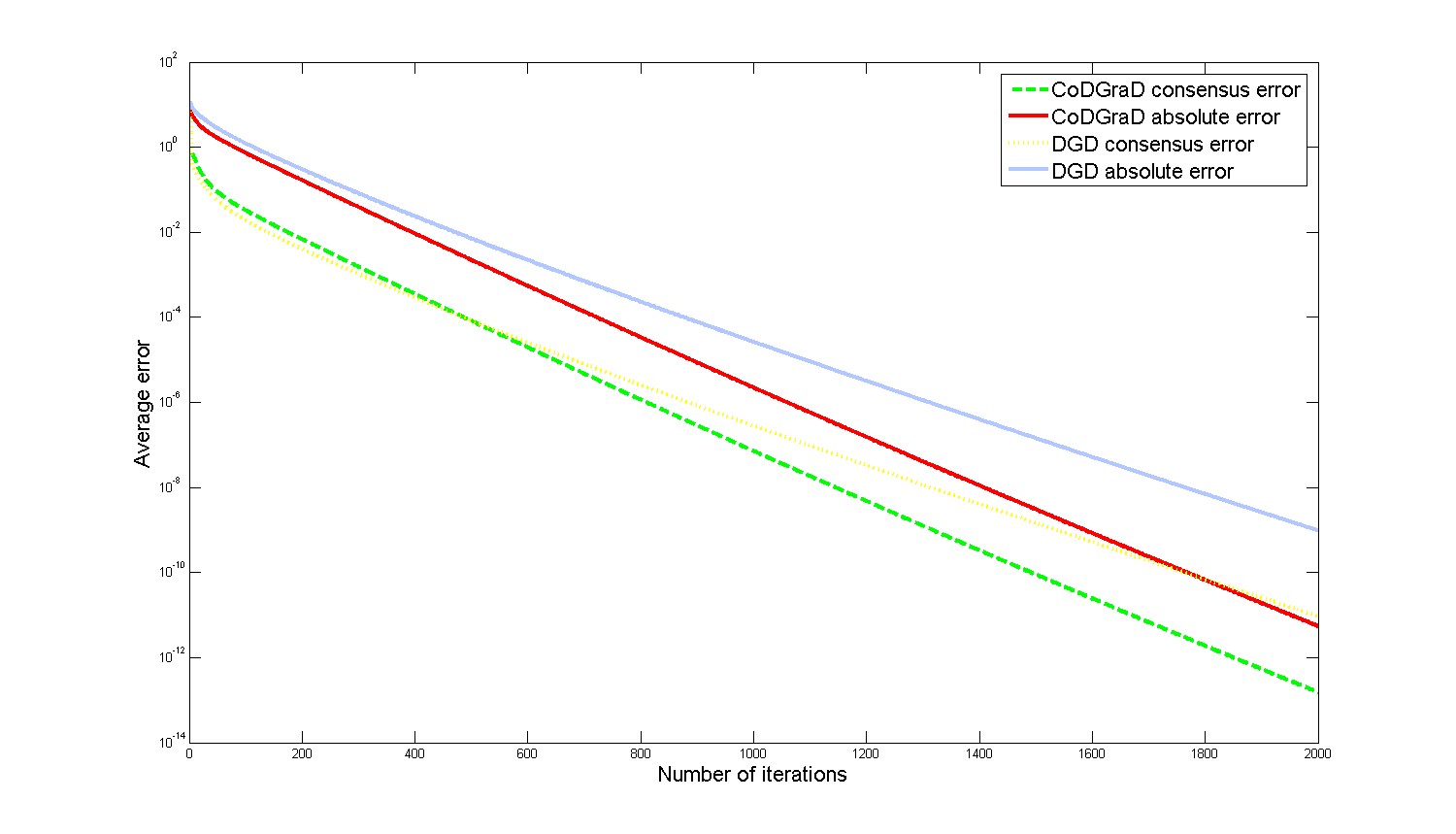}
\includegraphics  [width=78mm, height=45mm]{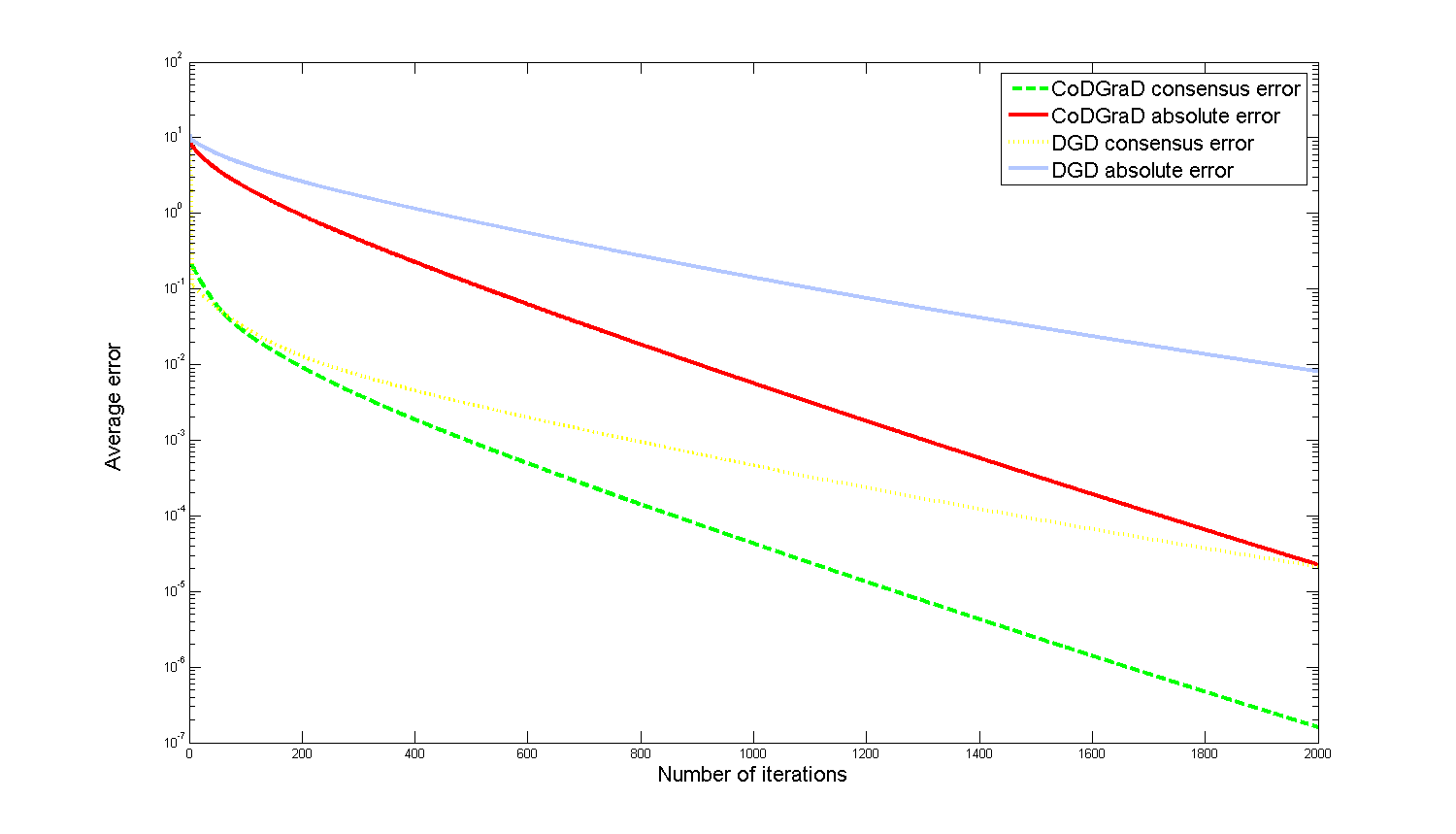}
\caption{Performance comparison of the CoDGraD  algorithm
\eqref{mainalgorithm} and the DGD algorithm \eqref{CTA}
over a five node network with $ (Q, N)= (225, 225) $ and step sizes $ \alpha_{k}=(k+800)^{-0.9}$ (top)
and $(k+500)^{-0.95}, k\ge 0 $ (bottom).
}
\label{fig_6}
\end{figure}

From the above simulations, we observe that the CoDGraD  algorithm
\eqref{mainalgorithm} has much better performance than the   CTA  algorithm \eqref{CTA}
in reaching consensus.
Even though
$\bar {\bf x} (k)$ satisfies a gradient descent algorithm \eqref{convergence.thm.pf.eq2} with
an inexact global gradient, see Proposition \ref{inexactgradient.prop},
our simulations indicate that the CoDGraD  algorithm
\eqref{mainalgorithm}
 still has comparable performance in the absolute error
with the  DGD algorithm under the CTA prototype \eqref{CTA}.
Also we  can conceive from the simulations that
the CoDGraD  algorithm
\eqref{mainalgorithm}
has faster convergence for a smaller exponent $\theta\in (1/2, 1]$, which confirms its convergence rate estimate in
 \eqref{example.estimate1} and \eqref{example.estimate2}.
On the other hand, our simulations also indicate that  decreasing the exponent $\theta$ moves the CoDGraD  algorithm
\eqref{mainalgorithm}  into the instability phase, which could directly be related to the sparsity of the network (i.e.,  the graph degree of the corresponding network and
the  degree distribution of vertices).
 It is  worth mentioning that we can adequately calibrate this instability by increasing the value of $ a $
in our illustrative examples of step sizes \eqref{stepsize.example}
 for a fixed exponent $ \theta $.
 Thus we  can anticipate in Figure  
  \ref{fig_4} that the increase in the value to $ \theta=0.85 $ degraded the convergence so that the  CoDGraD algorithm became closer in performance to the DGD algorithm. While in Figure  
  \ref{fig_6} we realize that the lower value of $ \theta=0.75 $ is impermissible since  the CoDGraD algorithm will considerably enter the instability region while a higher value of $ \theta=0.9 $ favors a better convergence rate and the highest value of $ \theta=0.95 $ degraded the convergence again.

In these simulations, we have compared the performance of
 the CoDGraD  algorithm \eqref{mainalgorithm} and the DGD algorithm \eqref{CTA} over the described $3$-node and $5$-node networks with  subsystems on nodes being  overdetermined. Therefore, a least squares solutions on each node will correspond to the unique solution of a strongly convex function and will consequently correspond to the least squared solution of the whole network, that is, the unique solution of the strongly convex global function $ f $.
It is observed that the errors decrease significantly in $ 2000 $ iterations where CoDGraD outperforms DGD in reaching the unique minimizer (i.e., absolute error) and in reaching consensus. While for both algorithms the consensus error decreases at a higher rate than the absolute error meaning that the workers become closer in their estimates while they all drift towards the unique solution.
Our further simulations  indicate that  convergence behaviors of the CoDGraD  algorithm \eqref{mainalgorithm} and the DGD algorithm \eqref{CTA} depends directly on maximal condition number of matrices $ {\bf G}_{i}, 1\le i\le n$, cf. \eqref{fconvergence.1} and \eqref{convergence.cor.eq01} where  $A$ is closely related to the maximal condition number in the current setting.

\section{Conclusions}\label{conclusion.section}

In this paper, we  proposed  the Code-Based Distributed Gradient Descent algorithm \eqref{mainalgorithm} to solve a convex optimization problem over a large network with  some workers  being stragglers due
to the failure or heavy delay on computing or communicating.
The proposed algorithm is a  distributed version of gradient descent algorithm with inexact gradient updating, and it
 has better performance in  reaching consensus as we apply the row stochastic matrix associated with  the coding/decoding scheme.
The convergence rate of the proposed CoDGraD algorithm  depends on the topological structure of the network,
the second largest eigenvalue of row stochastic matrix in magnitude, and  the updating step sizes in the algorithm.
Moreover, our coding scheme does not necessarily comply with the conventional paradigm of decomposing the global convex function onto a summand of local convex functions and  hence our coding/decoding scheme may shed new light on distributed inexact (stochastic) gradient descent algorithms.
We  wish that this work on CoDGraD will serves as a starting point for a full-fledged investigation on static and time-varying networks especially in the field of federated decentralized learning that we would like to continue resolving in the coming future.

\bigskip

{\bf Acknowledgement}: The authors would like to thank all reviewers for their constructive comments for the improvement of the manuscript.
 This work is partially supported by  the National Science Foundation (DMS-1816313).

\begin{appendix}

\subsection{Proof of  Proposition \ref{eigenvalueone.pr}}
\label{eigenvalueone.pr.appendix}

Observe that
$\tilde {\bf A}_++\tilde {\bf A}_+=|{\tilde {\bf A}}|$. Then, in order to establish
the equivalence in the proposition, it suffices to prove that
for any   positive integer $k\ge 1$ and nonzero complex number $\lambda$, the null space
of $\left(\left(\begin{array}{cc}
 {\bf A}  &   {\bf B} \\
 {\bf A} &  {\bf  B}
\end{array}\right)-
 \lambda {\bf I}_{2n} \right)^k $
and the one of $({\bf A}+{\bf B}- \lambda {\bf I}_{n})^k$, to be denoted by
${\mathcal N}_{k}^1(\lambda)$ and ${\mathcal N}_{k}^2(\lambda)$ respectively,
have the same dimension, where
 ${\bf A}$ and ${\bf B}$ are two square matrices of size $n\times n$.

  Set  ${\bf C}={\bf A}+{\bf B}$.
  By induction on $j\ge 1$,  we can show  that
\begin{equation}\label{eigenvalue.pf.eq1}
\left(\begin{array}{cc}
 {\bf A}  &   {\bf B} \\
 {\bf A} &  {\bf  B}
\end{array}\right)^j=
\left(\begin{array}{cc}
 {\bf C}^{j-1}  &   {\bf 0}_{n\times n} \\
 {\bf 0}_{n\times n} &  {\bf C}^{j-1}
\end{array}\right)
\left(\begin{array}{cc}
 {\bf A}  &   {\bf B} \\
 {\bf A} &  {\bf  B}
\end{array}\right).
\end{equation}
Therefore any
${\bf u}\in {\mathcal N}_{k}^2(\lambda)$, i.e.,
$({\bf C}- {\bf I}_{n})^k {\bf u}={\bf 0}_{n\times 1}$, we have
\begin{eqnarray*}
& \hskip-0.08in  & \hskip-0.08in
\left(
\left(\begin{array}{cc}
 {\bf A}  &   {\bf B} \\
 {\bf A} &  {\bf  B}
\end{array}\right)- \lambda
 {\bf I}_{2n} \right)^k
\left(\begin{array}{c}{\bf u}\\{\bf u}\end{array}\right)\\
& \hskip-0.08in = & \hskip-0.08in
\sum_{j=1}^k \binom {k}{j} (-\lambda)^{k-j} \left(\begin{array}{cc}
 {\bf A}  &   {\bf B} \\
 {\bf A} &  {\bf  B}
\end{array}\right)^j
\left(\begin{array}{c}{\bf u}\\{\bf u}\end{array}\right)
+ (-\lambda)^k \left(\begin{array}{c}{\bf u}\\{\bf u}\end{array}\right)
\\
& \hskip-0.08in = & \hskip-0.08in
 \left(\begin{array}{c} ({\bf C}-\lambda {\bf I}_{n})^k {\bf u}\\
({\bf C}-\lambda {\bf I}_{n})^k {\bf u}\end{array}\right) =
 \left(\begin{array}{c} {\bf 0}_{n\times 1}\\
 {\bf 0}_{n\times 1}\end{array}\right),
 \end{eqnarray*}
 where the second equality follows from \eqref{eigenvalue.pf.eq1}.
This proves that
   \begin{equation} \label{eigenvalue.pf.eq2}
   {\mathcal N}_{k}^1(\lambda)\supset\left\{
   \left(\begin{array}{c}{\bf u}\\{\bf u}\end{array}\right), \  {\bf u}\in {\mathcal N}_{k}^2(\lambda)\right\}.
  \end{equation}

On the other hand,  for any ${\bf u}, {\bf w}\in {\mathbb C}^n$ satisfying
\begin{equation*}
\left(
\left(\begin{array}{cc}
 {\bf A}  &   {\bf B} \\
 {\bf A} &  {\bf  B}
\end{array}\right)- \lambda
 {\bf I}_{2n} \right)^k
\left(\begin{array}{c}{\bf u}\\{\bf w}\end{array}\right)=
\left(\begin{array}{c} {\bf 0}_{n\times 1}\\{\bf 0}_{n\times 1}\end{array}\right),
    \end{equation*}
we obtain from \eqref{eigenvalue.pf.eq1} that
\begin{eqnarray*}
 & \hskip-0.1in  & \hskip-0.08in
 (-\lambda)^k \left(\begin{array}{c}{\bf u}\\{\bf w}\end{array}\right)
  =
-\sum_{j=1}^k \binom {k}{j} (-\lambda)^{k-j} \left(\begin{array}{cc}
 {\bf A}  &   {\bf B} \\
 {\bf A} &  {\bf  B}
\end{array}\right)^j
\left(\begin{array}{c}{\bf u}\\{\bf w}\end{array}\right)\\
& \hskip-0.18in= & \hskip-0.08in -\sum_{j=1}^k \binom {k}{j} (-\lambda)^{k-j}
\left(\begin{array}{c}
 {\bf C}^{j-1} ({\bf A} {\bf u} +{\bf B}{\bf w})  \\  {\bf C}^{j-1} ({\bf A} {\bf u} +{\bf B}{\bf w}) \end{array}\right).
 \end{eqnarray*}
This  implies that ${\bf w}={\bf u}$. Substituting the above equality back, we get
\begin{equation*}
 (-\lambda)^k \left(\begin{array}{c}{\bf u}\\{\bf u}\end{array}\right)
= -\sum_{j=1}^k \binom {k}{j} (-\lambda)^{k-j}
\left(\begin{array}{c}
 {\bf C}^{j}  {\bf u}   \\
 {\bf C}^{j} {\bf u}  \end{array}\right)
 \end{equation*}
 which implies that
 $({\bf C}-\lambda {\bf I}_{n})^k {\bf u}={\bf 0}_{n\times 1}$.
 Hence
    \begin{equation} \label{eigenvalue.pf.eq3}
   {\mathcal N}_{k}^1(\lambda)\subset\left\{
   \left(\begin{array}{c}{\bf u}\\{\bf u}\end{array}\right), \  {\bf u}\in {\mathcal N}_{k}^2(\lambda)\right\}.
  \end{equation}
 Combining \eqref{eigenvalue.pf.eq2} and \eqref{eigenvalue.pf.eq3}  
  completes the proof.

\subsection{Proof of Proposition \ref{tildezk.prop}}\label{tildezk.thm.appendix}
Denote the spectrum of a square matrix ${\bf A}$ by $\sigma({\bf A})$.
By the assumption on the  matrix
 ${\bf A}_{\rm sde}$, its spectrum  $\sigma({\bf A}_{\rm sde})$ satisfies
 \begin{equation}
\label{spectral.prop.eq3}
\sigma({\bf A}_{\rm sde})\subset \{1\} \cup \{z, \ |z|<1\},
\end{equation}
and
the eigenspace  associated with eigenvalue one is  given by
\begin{equation} \label{spectral.prop.eq4}
N({\bf A}_{\rm sde}-{\bf I})= {\rm span} \left\{
{\bf 1}_{2n}
\right\}.
\end{equation}
 Combining \eqref{P0.def}, \eqref{QP.identity}, \eqref{spectral.prop.eq3} and  \eqref{spectral.prop.eq4}, we obtain
 that the spectrum of ${\bf A}_{\rm sde}-{\bf P}$ is contained in the open unit disk,
\begin{equation}\label{QP.spectral}
\sigma({\bf A}_{\rm sde}-{\bf P})  =  \big(\sigma({\bf A}_{\rm sde})\backslash \{1\}\big)\cup \{0\} \subset
   \{z, \ |z|\le |\lambda_2({\bf A}_{\rm sde})|\}.
\end{equation}
Therefore there exists a positive constant $C_1$ such that
\begin{equation}\label{Asde-pk}
\|({\bf A}_{\rm sde}-{\bf P})^k\|_{{\mathcal B}^\infty}\le C_1\Big( \frac{1+2|\lambda_2({\bf A}_{\rm sde})|}{3}\Big)^k, \ k\ge 1,
\end{equation}
 where   $\|{\bf A}\|_{{\mathcal B}^\infty}= \sup_{\|{\bf x}\|_\infty=1} \|{\bf A}{\bf x}\|_\infty$.

By \eqref{Bfit.def}, \eqref{f.bounded}, \eqref{hi.def} and \eqref{zh.def2}, we have
\begin{equation}\label{h.estimate1}
\|{\bf h}(k)\|_{2, \infty} \le   \sup_{1\le i\le n} \sup_{{\bf x}\in {\mathbb R}^N} \|\nabla g_i({\bf x})\|_2\le M.
\end{equation}
 Then combining \eqref{tildezk.eq0},  \eqref{Asde-pk} and  \eqref{h.estimate1}  completes the proof.

\subsection{Proof of  Lemma  \ref{weigh.lem}}
\label{weigh.appendix}
By \eqref{QP.identity}, we have
$ ({\bf A}_{\rm sde}-{\bf P})^k= {\bf A}_{\rm sde}^k -{\bf P}$, $k\ge 1$.
Therefore
$
\lim_{k\to \infty} {\bf A}_{\rm sde}^k=
{\bf P}
$
by \eqref{Asde-pk}.  This, together with
 the observation that all entries of ${\bf A}_{\rm sde}^k, k\ge 1$ are nonnegative.
 Hence the required estimate
 implies that all entries of  ${\bf a}_{\rm sde}$ are nonnegative and the desired estimate  
  on weights
 follows.

\subsection{Proof of  Proposition  \ref{inexactgradient.prop}}
\label{inexactgradient.thm.appendix}
By  \eqref{hi.def}, \eqref{xypm.def} and \eqref{Lipschitz.con},   we have
\begin{eqnarray*}
     \| {\bf h}_i(k)-\nabla g_i(\bar {\bf x}(k))\|_2  
       & \hskip-0.08in =& \hskip-0.08in
    \| \nabla g_i({\bf z}_i(k))-\nabla g_i(\bar {\bf x}(k))\|_2 \nonumber\\  & \hskip-0.08in \le & \hskip-0.08in
     L  \|{\bf z}_i(k)-\bar{\bf x}(k)\|_2\le
 L  \|\tilde {\bf z}(k) 
 \|_{2, \infty}
\end{eqnarray*}
for $1\le i\le n$,
and
 $$ \| {\bf h}_{i}(k)+\nabla g_{i-n}(\bar {\bf x}(k))\|_2
 \le L  \|\tilde {\bf z}(k) 
 \|_{2, \infty} $$
 for $n+1\le i\le 2n$.
Therefore
\begin{equation}\label{htildeh.est}
\| {\bf h}(k) -\tilde {\bf h}(k) \|_{2, \infty}\le L
 \|\tilde {\bf z}(k) 
 \|_{2, \infty},
\end{equation}
 where
$$\tilde {\bf h}(k)=\left( \begin{array}{c} \nabla G(\bar {\bf x}(k))\\
 -\nabla G(\bar {\bf x}(k))\end{array}\right)\ \ {\rm and} \ \
 \nabla G({\bf x})=
\left( \begin{array}{c} \nabla g_1({\bf x})\\
\vdots\\
\nabla g_n( {\bf x})\end{array}\right).
$$
 Therefore
\begin{eqnarray}\label{upesilon.eq11}
\hskip-0.08in &  \hskip-0.08in  & \hskip-0.08in  \|{\bf a}_{\rm sde}^T  {\bf h}(k)- {\bf a}_{\rm sde}^T \tilde {\bf h}(k)\|_2
 \nonumber\\
 \hskip-0.08in & \hskip-0.08in \le &  \hskip-0.08in  \|{\bf a}_{\rm sde}\|_1 
\|{\bf h}(k)- \tilde {\bf h}(k)\|_{2, \infty} \le L \|{\bf a}_{\rm sde}\|_1 \|\tilde {\bf z}(k) 
\|_{2, \infty},  \qquad
\end{eqnarray}
where
the second  estimate  follows from \eqref{htildeh.est}.

 Observe from  \eqref{AB1.eq} that
 $ {\bf A}_{\rm sde} \tilde {\bf h}(k)=
 \nabla f (\bar{\bf x}(k)) {\bf w}$.
This together with \eqref{asde.def} implies that
\vskip-0.16in
  \begin{equation}\label{uqtildeh.eq}
 {\bf a}_{\rm sde} ^T \tilde {\bf h}(k)=
 {\bf a}_{\rm sde} ^T  {\bf A}_{\rm sde}
  \tilde {\bf h}(k)=\tilde w \nabla f(\bar{\bf x}(k)).\end{equation}
  \vskip-0.06in
  \noindent
Combining \eqref{upesilon.eq11} and \eqref {uqtildeh.eq}  proves the desired estimate
\eqref{reminder.estimate}.

\subsection{Proof of  Proposition \ref{preconvergence.prop}}
\label{preconvergence.thm.appendix}

Set
\begin{equation}\label{beta.def1} \beta_k= M^2  \alpha_k^2+ L^2 \|\tilde {\bf z}(k)\|_{2, \infty}^2, \  k\ge 0.\end{equation}
 By \eqref{hi.def}, \eqref{f.bounded}, \eqref{convergence.thm.pf.eq2}, \eqref{reminder.estimate} and
\eqref{htildeh.est},
 we obtain
 \vskip-0.16in
\begin{eqnarray}\label{convergence.thm.pf.eq01}
\|\bar{\bf x}(k+1)-\bar {\bf x}\|_2^2   & \hskip-0.08in  = &  \hskip-0.08in  \|\bar{\bf x}(k)-\bar {\bf x}\|_2^2+ \alpha_k^2 \|{\bf a}_{\rm sde}^T  {\bf h}(k)\|_2^2\nonumber \\
\hskip-0.08in  & \hskip-0.08in  &  \hskip-0.08in   -
2 \alpha_k \langle {\bf a}_{\rm sde}^T   {\bf h}(k), \bar{\bf x}(k)-\bar {\bf x}\rangle_N\nonumber\\
 & \hskip-0.08in  = &  \hskip-0.08in  \|\bar{\bf x}(k)-\bar {\bf x}\|_2^2+ \alpha_k^2 \|{\bf a}_{\rm sde}^T  {\bf h}(k)\|_2^2\nonumber \\
  & \hskip-0.01in &  \hskip-0.01in - 2 \alpha_k \langle {\bf a}_{\rm sde}^T {\bf h}(k) - {\bf a}_{\rm sde}^T   \tilde{\bf h}(k), \bar{\bf x}(k)-\bar {\bf x}\rangle_N   \nonumber \\
  & \hskip-0.08in  &  \hskip-0.08in - 2 \alpha_k \langle {\bf a}_{\rm sde}^T   \tilde{\bf h}(k), \bar{\bf x}(k)-\bar {\bf x}\rangle_N
\nonumber\\
& \hskip-0.08in  \le  &  \hskip-0.08in   \|\bar{\bf x}(k)-\bar {\bf x}\|_2^2 + M^2 \|{\bf a}_{\rm sde}\|_1^2 \alpha_k^2 \nonumber\\
 & \hskip-0.08in    &  \hskip-0.08in   + 2 L  \|{\bf a}_{\rm sde}\|_1 \alpha_k \|\tilde {\bf z}(k)\|_{2, \infty} \|\bar{\bf x}(k)-\bar {\bf x}\|_{2}\nonumber\\
 & \hskip-0.08in  \le  &  \hskip-0.08in  (1+ \alpha_k^2 )\|\bar{\bf x}(k)-\bar {\bf x}\|_2^2+ \beta_k,
\end{eqnarray}
\vskip-0.06in
\noindent
where we also use  the positivity of $\tilde w$ in Lemma \ref{weigh.lem}, $ \|{\bf a}_{\rm sde}\|_1 = 1 $  and the convexity of the objective function $f$,
\vskip-0.16in
\begin{equation}\label{convex.basic}
\langle \nabla f({\bf x}), {\bf x}-\bar{\bf x}\rangle_N\ge 0,  \ {\bf x}\in {\mathbb R}^N.
\end{equation}
\vskip-0.06in
\noindent
Applying (\ref{convergence.thm.pf.eq01}) repeatedly, we get
\vskip-0.16in
\begin{eqnarray}\label{convergence.thm.pf.eq02}
   & \hskip-0.08in  & \hskip-0.08in    \|\bar{\bf x}(k+1)-\bar {\bf x}\|_2^2 \nonumber  \\
   & \hskip-0.08in \le & \hskip-0.08in    \prod_{j=0}^k (1+\alpha_j^2) \|\bar{\bf x}(0)-\bar {\bf x}\|_2^2  +\beta_k+ \sum_{j=0}^{k-1} \beta_j \prod_{j'=j+1}^k (1+\alpha_{j'}^2)\nonumber\\
   &  \le  & \exp\Big(\sum_{j=0}^{k} \alpha_j^2\Big) \Big(\|\bar{\bf x}(0)-\bar {\bf x}\|_2^2 +\sum_{j=0}^k \beta_j\Big)\nonumber\\
   & \le & \exp\Big(\sum_{j=0}^{\infty} \alpha_j^2\Big) \Big(\|\bar{\bf x}(0)-\bar {\bf x}\|_2^2 +\sum_{j=0}^\infty \beta_j\Big),\ \  k\ge 1.
%
\end{eqnarray}
\vskip-0.06in
\noindent
%
This together with \eqref{alphak.assumption}  and \eqref{tildezk.est001} proves the desired bound  \eqref{preconvergence.thm.eq1}
for the sequence $\bar {\bf x}(k), k\ge 0$.

\vskip-0.16in

\subsection{Proof of  Proposition \ref{convergence.prop}}
\label{convergence.thm.appendix}

By \eqref{alphak.assumption}, without a loss of generality, we assume that
$0\le \alpha_k\le  \tilde w A$ for all $k\ge 0$.
Then for $k\ge 1$,  following the argument in \eqref{convergence.thm.pf.eq01} with  the convexity \eqref{convex.basic} replaced by
the strong  convexity \eqref{accretive.con}, we obtain
that
\vskip-0.16in
$$ \|\bar{\bf x}(k)-\bar {\bf x}\|_2^2
\le   (1- \tilde w  A \alpha_{k-1})\|\bar{\bf x}(k-1)-\bar {\bf x}\|_2^2
 +\beta_{k-1},$$
\vskip-0.06in
\noindent cf.  \eqref{convergence.thm.pf.eq01}.
Applying the above estimate repeatedly leads to
\vskip-0.16in
\begin{eqnarray*} 
 \hskip-0.08in \|\bar{\bf x}(k)-\bar {\bf  x}\|_2^2   & \hskip-0.08in \le & \hskip-0.08in   \prod_{j=0}^{k-1} (1- \tilde w  A \alpha_j) \|\bar{\bf x}(0)-\bar {\bf x}\|_2^2 + \beta_{k-1}\nonumber\\
& \hskip-0.08in  & \hskip-0.08in  + \sum_{j=0}^{k-2} \beta_j \prod_{j'=j+1}^{k-1} (1-\tilde w  A\alpha_{j'})\nonumber\\
& \hskip-0.08in \le & \hskip-0.08in   \exp\Big(-\tilde w A\sum_{j=0}^{k-1} \alpha_j\Big) \|\bar{\bf x}(0)-\bar {\bf x}\|_2^2 + \beta_{k-1}\nonumber\\
& \hskip-0.08in  & \hskip-0.08in  + \sum_{j=0}^{k-2} \exp\Big(-\tilde w  A\sum_{j=j+1}^{k-1} \alpha_j\Big) \beta_j,
\end{eqnarray*}
\vskip-0.06in
\noindent where $\beta_j, j\ge 0$, is given in \eqref{beta.def1}.
This proves the desired estimate \eqref{convergence.thm.eq01}.

\end{appendix}

\begin{IEEEbiography}
    [{\includegraphics[width=1in,height=1in,clip,keepaspectratio]{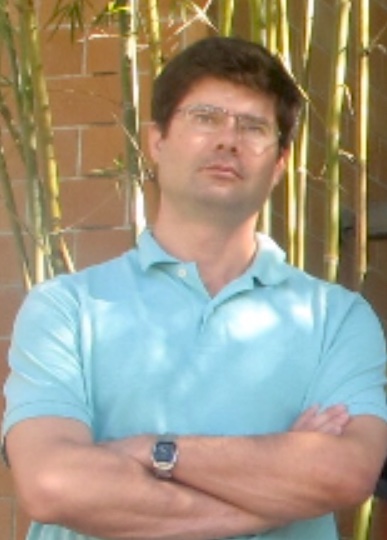}}]{Elie Atallah} received his Ph.D degree in Electrical Engineering from the University of Central Florida in 2019, an M.S degree in Electrical Engineering and an M.S degree in Mathematics both from University of California, Riverside in 2003 and 2006, respectively. At UCF his primary focus was in developing algorithms for distributed optimization, Compressive Sensing and Tensor Decomposition. Previously, he has been an adjunct faculty at several universities such as California Baptist University, Seminole State College, Valencia College and UCF.
\end{IEEEbiography}

\begin{IEEEbiography}
    [{\includegraphics[width=1in,height=1.25in,clip,keepaspectratio]{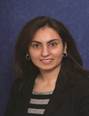}}]{Nazanin Rahnavard} (S’97-M’10, SM’19) received her Ph.D. in the School of Electrical and Computer Engineering at the Georgia Institute of Technology, Atlanta, in 2007. She is currently an Associate Professor in the Department of Electrical and Computer Engineering at the University of Central Florida, Orlando, Florida. Dr. Rahnavard is the recipient of NSF CAREER award in 2011. She has interest and expertise in a variety of research topics in the communications, networking, and signal processing areas. She serves on the editorial board of the Elsevier Journal on Computer Networks (COMNET) and on the Technical Program Committee of several prestigious international conferences.
\end{IEEEbiography}

\begin{IEEEbiography}
    [{\includegraphics[width=1in,height=1.25in,clip,keepaspectratio]{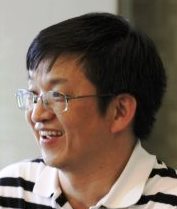}}]{Qiyu Sun} received the Ph.D. degree in Mathematics from Hangzhou University, Hangzhou, China, in
1990. He is currently a Professor of Mathematics at the University of Central Florida, Orlando, FL,
USA. 
 His research interests include applied and computational harmonic analysis, sampling theory,
phase retrieval and graph signal processing. He has published more than 120 papers. 
 He received the 2019 Best SICON Paper Prize, presented by the Society for Industrial and Applied Mathematics (SIAM) Activity Group on Control and Systems Theory (SIAG/CST).
\end{IEEEbiography}

\end{document}